\documentclass[11pt,a4paper]{article}
\usepackage{amsfonts,amssymb}
\usepackage{cite}
\usepackage[T2A]{fontenc}
\usepackage[cp1251]{inputenc}
\usepackage[english]{babel}
\usepackage{amsmath}%
\usepackage[unicode]{hyperref}
\usepackage[left=1.5cm,right=1cm, top=2cm,bottom=2cm,bindingoffset=0cm]{geometry}

\newtheorem{theorem}{Theorem}[section]
\newtheorem{lemma}{Lemma}[section]
\newtheorem{conjecture}{Conjecture}[section]
\newtheorem{proof}{Proof}[section]
\newtheorem{example}{Example}[section]
\newtheorem{corollary}{Corollary}[section]

\allowdisplaybreaks

\title{\bf Characteristic Lie Algebras \\ of Integrable Differential-Difference Equations in 3D}
\author{\bf I. Habibullin and A. Khakimova}

\date{}
\begin{document}
\maketitle

\abstract{The purpose of this article is to develop an algebraic approach to the problem of integrable classification of differential-difference equations with one continuous and two discrete variables. As a classification criterion, we put forward the following hypothesis. Any integrable equation of the type under consideration admits an infinite sequence of finite-field Darboux-integrable reductions. The property of Darboux integrability of a finite-field system is formalized as finite-dimensionality condition of its characteristic Lie-Rinehart algebras. That allows one to derive effective integrability conditions in the form of differential equations on the right hand side of the equation under study. To test the hypothesis, we use known integrable equations from the class under consideration. In this article, we show that all known examples do have this property.}

\maketitle

\section{Introduction}

The theory of integrability is an important component of modern mathematical physics. At present, the problem of classifying integrable equations with three and more independent variables is being actively studied. For these classes of equations, in contrast to the case of two variables, there are several approaches to the problem of integrable classification (see \cite{Bogdanov, BogdanovKonopelchenko, DavidCalderbank, CalderbankKruglikov, CleryFerapontov, DoubrovFerapontovKruglikovNovikov2018, DoubrovFerapontovKruglikovNovikov2019, DunajskiMasonTod, FerapontovHadjikosKhusnutdinova, FerapontovKhusnutdinova, FerapontovKruglikov, FerapontovMoro, FerapontovMoroNovikov, FerapontovOdesskii, FNR, GibbonsKodama, GibbonsTsarev, Hitchin, Krichever, OdesskiiSokolov, Pavlov2003, Pavlov2004, Penrose, Ward, Zakharov}) mostly based on geometric ideas and constructions. Note that the well-known generalized symmetry method is not efficient in higher dimensions due to the non-locality problem.

The hydrodynamic reduction method is one of the most widespread among specialists in mathematical physics. Let us briefly explain the essence of this method for the classification of three-dimensional integrable equations, based on the following two-stage procedure. The first consists in classification of dispersionless integrable systems by using the method of hydrodynamic reductions (initiated by Gibbons, Kodama and Tsarev (see \cite{GibbonsKodama, GibbonsTsarev}) and developed into an efficient integrability test by the group of Ferapontov, Khusnutdinova, Novikov, Pavlov, Odesskii and Sokolov (\cite{FerapontovKhusnutdinova, OdesskiiSokolov, Pavlov2003, Pavlov2004}). Second is reconstruction of dispersive deformations of dispersionless integrable systems based on the method of dispersive deformations of hydrodynamic reductions developed by Ferapontov, Moro and Novikov (\cite{FerapontovMoro, FerapontovMoroNovikov}).

The article discusses a problem of developing an  algebraic integrability criterion  based on the Darboux integrable  reductions for the differential-difference equations of the form:
\begin{align} \label{eq0}  
u_{n+1,x}^j = F(u_{n,x}^j,u_n^{j+1},u_{n+1}^j,u_n^j,u_{n+1}^{j-1}), \quad -\infty < n,j < \infty,
\end{align}
where the sought function $u$ depends on three variables, on a real $x$ and two integers $j$ and $n$. We assume that 
\begin{align*}
\frac{\partial F}{\partial u_{n,x}^j}\neq 0, \quad \frac{\partial F}{\partial u_n^{j+1}}\neq 0, \quad \frac{\partial F}{\partial u_{n+1}^{j-1}}\neq 0. 
\end{align*}
The lattice is defined on a quadrilateral graph. 
To preserve the parity of the forward and backward directions in $n$ we assume that equation (\ref{eq0}) can be uniquely rewritten as
\begin{align*}
u_{n-1,x}^j=G(u_{n,x}^j,u_{n-1}^{j+1},u_{n}^j,u_{n-1}^j,u_{n}^{j-1}).
\end{align*}

Integrable equations of the form (\ref{eq0}) have been studied by many authors (see \cite{FNR} and the references therein). Note that, in some particular cases, equations of the form (\ref{eq0}) can be obtained as the Backlund transformation for the two-dimensional Toda-type chains
\begin{align}  \label{eq0B}  
u_{n,xy} = f(u_{n+1},u_n,u_{n-1}, u_{n,x},u_{n,y}), \quad -\infty < n < \infty.
\end{align}

Earlier the lattices (\ref{eq0}) have been investigated by means of the over mentioned method of hydrodynamic reductions \cite{FNR}. Within the framework of this approach, a certain class of lattices of the form (\ref{eq0}) is studied and integrable cases are selected.
 
In the studies \cite{HabibullinKuznetsova, HabibullinKuznetsovaSakieva, HabPoptsovaSIGMA17, HabPoptsovaUMJ} a classification algorithm was developed for integrable equations with three independent variables  based on the idea of the Darboux integrable reductions and characteristic Lie algebras. Efficiency of the method was illustrated by applying to the equations from the class (\ref{eq0B}) (see \cite{HabibullinKuznetsova}).

It was observed in the article \cite{FHKN} that it is very convenient to study the classification problem for (\ref{eq0B}) by combining different approaches.  At the first step,   equations are selected that satisfy the requirement  that the dispersionless limit of the equation is integrable, that is  its characteristic variety defines a conformal structure, which is Einstein-Weyl, on every solution. As a result one obtaines a class of equations having the  freedom in only one function of one variable and a few constant parameters. And then, to the obtained equations from this class  the test of Darboux integrable reductions is applied.

In this paper, we show at the level of examples that integrable  differential-difference lattices of the form (\ref{eq0}) admit Darboux integrable reductions.  More precisely, the purpose of the article is to approve the following 

\begin{conjecture} \label{conjecture} A lattice of the form (\ref{eq0}) is integrable if and only if there exists a pair of functions $H^{(1)}$ and $H^{(2)}$ of four variables such that for any choice of the integer $N$, a system of the hyperbolic type differential-difference equations
\begin{align}
&u_{n+1,x}^1 =H^{(1)}(u_{n,x}^1,u_n^{2},u_{n+1}^1,u_n^1), \nonumber \\
&u_{n+1,x}^j =F^j(u_{n,x}^j,u_n^{j+1},u_{n+1}^j,u_n^j,u_{n+1}^{j-1}), \quad 1 < j < N, \label{eq00} \\
&u_{n+1,x}^N =H^{(2)}(u_{n,x}^N,u_{n+1}^N,u_n^N,u_{n+1}^{N-1}), \quad -\infty < n < \infty \nonumber
\end{align}  
obtained from (\ref{eq0}) is integrable in the sense of Darboux.
\end{conjecture}

Let us comment on the essence of the Conjecture~\ref{conjecture}. Finite-field reductions of integrable lattices that inherit integrability are quite common, for example, they are used to construct particular solutions of chains.  Such reductions are provided by imposing boundary conditions (cutoff conditions) compatible with the integrability property. The most popular of this kind conditions is the periodicity closure constraint. Usually, the finite-field reductions obtained as a result of integrable cutoffs are soliton systems, and only in the exceptional case, which is obtained by imposing degenerate cutoff conditions, do we arrive at systems integrable in the sense of Darboux. Our hypothesis is that, on the one hand, any integrable chain of the form (\ref{eq0}) admits degenerate boundary conditions, which reduce the chain to a Darboux integrable system, and on the other hand, a chain admitting degenerate boundary conditions with this property, certainly, is integrable. Existence of Darboux integrable reductions allows one to apply the technique of the characteristic Lie-Rinehart algebras for the classifying the lattices (\ref{eq0}). 

The systematic study of partial differential equations of hyperbolic type of the form
\begin{align*}
u_{x,y}=f(x,y,u,u_x,u_y),
\end{align*} 
the integration problem of which is reduced to solving ordinary differential equations began apparently in the second half of the nineteenth century. The problem of a complete description of this class of equations, called Darboux integrable equations or equations of Liouville type   was studied in the well-known works of E. Goursat. The task of complete classification turned out to be very difficult, and later many researchers worked in this field. The reader can find an overview of the results and a description of the current state of the problem in a detailed article \cite{ZhS2001}.

Systems of nonlinear hyperbolic equations remain less studied from the point of view of Darboux integrability. The exception is a class of the so-called exponential type systems \cite{GanzhaTsarev,Goursat,Leznov,Mikhailov,ShabatYamilov}
\begin{align*}
u^i_{x,y}=\exp \left(\sum_j a_{i,j}u^j\right).
\end{align*}
It is known that this system is Darboux integrable if and only if $A=\{a_{i,j}\}$ is the Cartan matrix of a semi-simple Lie algebra. It is remarkable that general solution for Darboux integrable systems can be constructed in a closed form (see \cite{Darboux,GanzhaTsarev,Leznov,ZhV2013,ZhS2001}).  

Characteristic algebras for the Darboux integrable scalar differential-difference equations of the form
\begin{align*}
u_{n+1,x}=u_{n,x}+f(u_n,u_{n+1})
\end{align*}
are introduced and investigated in \cite{HPZ-TJM2008,HPZ2008,HPZ2009}. In \cite{HPZ2009} the complete list of such equations is given. Darboux integrable systems of differential-difference equations of the exponential type are studied in \cite{HabZhYan2011} and \cite{Smirnov2015}.

We briefly discuss the content of the article. In \S 2 we consider finite systems of the differential-difference equations of the form (\ref{eqG}). We recall the definition of the $x$- and $n$-integrals and define the integrability in the sense of Darboux. We introduce the characteristic vector fields and the characteristic Lie-Rinehart algebras in the directions of $n$ and $x$. In \S 3, we study the structure of characteristic algebras and describe the action of their automorphisms, which will be needed when solving the classification problem. We then in \S 4 illustrate how to use characteristic algebra to study the problem of integrable classification of the systems of the form (\ref{eq00}).
 As an example, we derive one of the necessary conditions of integrability for the system (\ref{eq00}). First, we derive a system of nonlinear functional equations from the requirement that the characteristic algebra is finite-dimensional, and then we reduce these functional equations to a system of differential equations. In Section 5, we study all seven known differential-difference integrable lattices with one continuous and two discrete independent variables. In each of these equations, we make a change of variables in order to obtain an equation that admits a degenerate cutoff. By imposing the cutoff conditions at two points, we reduce the equation to a finite-field system of differential-difference equations. Further, we show that the resulting finite-field systems are Darboux integrable by presenting for these systems complete sets of integrals in both characteristic directions (among them there are scalar equations, they are presented separately in section 5.8). This circumstance confirms the correctness of the approach to the classification problem formulated in the Conjecture~\ref{conjecture} above.  We use two different methods to construct integrals. In some cases, when the Lax pair satisfies the conditions of Lemma~\ref{LemP-x} or Lemma~\ref{LemRk-n}, it is convenient to use the Lax pair. In other cases, we use the method of characteristic Lie-Reinhart algebras. This method is always applicable for constructing integrals, although it is more labor-consuming.

\section{Darboux integrable systems of differential-difference  equations and characteristic Lie-Rinehart algebras}

In this section we concentrate on a system of differential-difference equations of general form
\begin{equation}  \label{eqG}  
u_{n+1,x}^j = F^j(x,n,u_{n,x}^j,u_n,u_{n+1}), \quad j=1,2,\ldots,N,
\end{equation}
where $u_n=(u^1_n,u^2_n,\ldots,u^N_n)$. We request that functions $F^j$ are analytic in a domain in $\bf C^{2N+1}$, where $\bf C$ is the complex plane.

We assume that system (\ref{eqG}) can be rewritten in the converse way
\begin{equation*}  \label{eqGinv}  
u_{n-1,x}^j = G^j(x,n,u_{n,x}^j,u_n,u_{n-1}), \quad j=1,2,\ldots,N.
\end{equation*}
Variables $u^j_n, u^j_{n\pm1}, u^j_{n \pm2}, \ldots$ and $u^j_{n, x}, u^j_{n, xx}, u^j_{n,xxx}, \ldots$ are called dynamical variables, we treat them as independent ones.

In what follows we will be interested in the cases when system (\ref{eqG}) admits nontrivial $x$- and $n$-integrals. Let us give the necessary definitions. We call a function $I=I(x,n, u,u_{n,x},u_{n,xx},\ldots)$ $n$-integral of (\ref{eqG}) if it satisfies the condition $D_nI=I$ by means of the system, where $D_n$ is the shift operator acting due to the rule $D_n y(n)=y(n+1)$. Similarly function $J=J(x,n,u_n,u_{n\pm1}, u_{n\pm2},\ldots)$ is an $x$-integral if the equation $D_x J=0$ holds, where $D_x$ stands for the operator of the total derivative with respect to $x$.
Integrals $I=I(x)$ and $J=J(n)$ depending only on $x$ and correspondingly only on $n$ are called trivial integrals. System (\ref{eqG}) is called integrable in sense of Darboux if it possesses the complete set of nontrivial integrals in both characteristic directions $x$ and $n$. Completeness means that the number of functionally independent integrals in each direction coincides with the order $N$ of the system.

In order to find the integrals one can use the so-called characteristic operators. Let us discuss how these operators are derived. Assume that $J$ is an $x$-integral, then according to the definition we have $D_x J(x,n,u_n,u_{n\pm1}, u_{n\pm2},\ldots)=0$, we use the chain rule and get
\begin{equation}\label{KJ}
K_0J=\left(\frac{\partial}{\partial x}+\sum_{j=1}^{p} u_{n,x}^j \frac{\partial}{\partial u_{n}^j}+u_{n+1,x}^j\frac{\partial}{\partial u_{n+1}^j}+ u_{n-1,x}^j \frac{\partial}{\partial u^j_{n-1}} +\ldots\right)J=0.
\end{equation}
We specify the operator $K_0$ by means of the system (\ref{eqG}) and obtain
\begin{equation}\label{K0}
K_0=\frac{\partial}{\partial x}+\sum_{j=1}^{p} u_{n,x}^j \frac{\partial}{\partial u_{n}^j}+F_{n}^j\frac{\partial}{\partial u_{n+1}^j}+ G_{n}^j \frac{\partial}{\partial u^j_{n-1}} +F_{n+1}^j\frac{\partial}{\partial u_{n+2}^j}+ G_{n-1}^j \frac{\partial}{\partial u^j_{n-2}} +\ldots,
\end{equation}
where $F^j_{n+i}=F^j(x,n+i,u_{n+i,x}^j,u_{n+i},u_{n+i+1})$ and $G^j_{n-i}= G^j(x,n-i,u_{n-i,x}^j,u_{n-i},u_{n-i-1})$. We call $K_0$ the characteristic operator in the $x$ direction.

An important feature of the equation (\ref{KJ}) is that the solution $J$ does not depend on the variables $u^1_{n,x}, u^2_{n,x},\ldots,u^N_{n,x}$ despite the fact that the coefficients of the equation depend on them. Therefore in addition to (\ref{KJ}) any $x$-integral solves also a set of the following equations
\begin{equation}\label{Xk}
X_kJ=0, \quad k=1,2,\ldots,N,
\end{equation}
where $X_k=\frac{\partial}{\partial u^k_{n,x}}$. Evidently $x$-integral is nullified by the  multiple commutators of the operators 
\begin{equation}\label{XkK}
X_1, X_2, \ldots, X_N, K_0
\end{equation}
as well as linear combinations of the obtained operators with variable coefficients, depending on the dynamical variables. In other words, any $x$-integral belongs to the kernel of the operators in the Lie-Rinehart algebra $L_x$, generated by the operators (\ref{XkK}) over the ring $A$ of locally analytic functions depending on the dynamical variables. We call it characteristic algebra in the direction of $x$. Actually algebra $L_x$ is a finitely generated module over the ring $A$. Here the elements in $L_x$ can be multiplied by functions from $A$. The consistency conditions:
\begin{itemize}
\item[1)] $[W_1,aW_2]=W_1(a)W_2+a[W_1,W_2]$,
\item[2)] $(aW_1)b=aW_1(b)$
\end{itemize} 
are assumed to be valid for any $W_1,W_2\in L_x$ and $a,b\in A$. In other words we request that, if $W_1\in L_x$ and $a\in A$ then $aW_1\in L_x$ (see \cite{Rinehart}). 

The algebra $L_x$ is of a  finite dimension if it admits a  basis consisting of a finite number of the elements $W_1,W_2,\ldots,W_k\in L_x$ such that an arbitrary operator $W\in L_x$ is represented as a linear combination of the form
\begin{equation}	\label{linear_comb}
W=a_1W_1+a_2W_2+\dots +a_kW_k,
\end{equation}
where the coefficients are functions $a_1, a_2, \ldots, a_k\in A$. If in (\ref{linear_comb}) $W=0$ then we have $a_1=0$, $a_2=0$, $\ldots$, $a_k=0$.

Let us now discuss the $n$-integrals. Assume that a function $I=I(x,n,u_n,u_{n,x},u_{n,xx},\ldots)$ is an $n$-integral for the system (\ref{eqG}), i.e. the following relation holds $D_nI=I$, or the same
\begin{equation*}
I(x,n+1,u_{n+1},u_{n+1,x},u_{n+1,xx},...)=I(x,n,u_n,u_{n,x},u_{n,xx},\ldots).
\end{equation*}
Due to the equation (\ref{eqG}), that can be written in a shortened way as $u_{n+1,x}=F(u_{n,x}, u_n,u_{n+1})$ where \linebreak $F=(F^1,F^2,\ldots,F^N)$  we rewrite the latter as
\begin{equation} \label{DnI}
I(x,n+1,u_{n+1},F,D_xF, D_x^2F,\ldots)=I(x,n,u_n,u_{n,x},u_{n,xx},\ldots).
\end{equation}
The problem of finding $n$-integrals is much more difficult than finding $x$-integrals, indeed, in this case, one needs to solve not a differential but a functional equation (\ref{DnI}). Below we introduce characteristic operators that allow to study completely this equation. It is easily observed that the right hand side of (\ref{DnI}) does not contain dependence on the variable $u_{n+1}=(u_{n+1}^1, u_{n+1}^2, \ldots,u_{n+1}^N)$ therefore the left hand side satisfies the relations $\frac{\partial}{\partial u^j_{n+1}}D_nI=0$, which evidently implies
\begin{equation} \label{derivative}
Y_{j,1}I:=D_n^{-1}\frac{\partial}{\partial u^j_{n+1}}D_nI=0\quad \mbox{for}\quad j=1,2,\ldots,N.
\end{equation}
It can be proved by a direct computation that characteristic operators defined in (\ref{derivative}) act on the dynamical variables $u_n,u_{n,x},u_{n,xx},\ldots$ as vector fields of the form
\begin{equation} \label{vectorfields}
Y_{j,1}=\frac{\partial}{\partial u^j_{n}}+\sum_{i=1}^N D_n^{-1}\left(\frac{\partial F_n^i}{\partial u^j_{n+1}}\right)\frac{\partial}{\partial u^i_{n,x}}
+D_n^{-1}\left(\frac{\partial F_{n,x}^i}{\partial u^j_{n+1}}\right)\frac{\partial}{\partial u^i_{n,xx}}+\dots,
\end{equation}
where $F_{n,x}^i:=D_xF_{n}^i$. We set $Y_{j,0}=\frac{\partial}{\partial u^j_{n+1}}$ and then rewrite representation as follows
\begin{equation*} \label{vectorfields1}
Y_{j,1}=\frac{\partial}{\partial u^j_{n}}+\sum_{i=1}^N D_n^{-1}\left(Y_{j,0}\left(F_n^i\right)\right)\frac{\partial}{\partial u^i_{n,x}}
+D_n^{-1}\left(Y_{j,0}\left(F_{n,x}^i\right)\right)\frac{\partial}{\partial u^i_{n,xx}}+\dots. 
\end{equation*}
In contrast to the case of the previously studied $x$-integrals, in this case  there are additional characteristic operators, which also annul  $I$.
For any natural number $k$, the following relation holds:
\begin{equation} \label{derivative-k}
Y_{j,k}I:=D_n^{-k}\frac{\partial}{\partial u^j_{n+1}}D_n^kI=0\quad \mbox{for}\quad j=1,2,\ldots,N.
\end{equation}
We give uniform representations for all characteristic operators (\ref{derivative-k}) for $k\geq2$
\begin{equation} \label{vectorfieldsk}
Y_{j,k}=\sum_{i=1}^N D_n^{-1}\left(Y_{j,k-1}\left(F_n^i\right)\right)\frac{\partial}{\partial u^i_{n,x}}
+D_n^{-1}\left(Y_{j,k-1}\left(F_{n,x}^i\right)\right)\frac{\partial}{\partial u^i_{n,xx}}+\dots. 
\end{equation}
Since the coefficients of the equations (\ref{derivative-k}) depend on the variables $u^j_{n-s}$ for $1\leq s\leq k$ and $1\leq j\leq N$, while the solution to the equation does not depend on them, then such equations 
\begin{equation*} \label{barX}
\bar X_{i,s}I=0,
\end{equation*}
where $\bar X_{i,s}=\frac{\partial}{\partial u^i_{n-s}}$ should hold for these values of $i$ and $s$.

It is important that the operators $Y_{j,i}$  commute with each other.

\begin{lemma}
For any $i, i'\geq 0$ and $j, j'$ such that $1\leq j, j'\leq N$ the relation $\left[Y_{j,i},Y_{j',i'}\right]=0$ holds.
\end{lemma}

\begin{proof}
Let us first assume that $i'=0$, then obviously 
\begin{equation*}
\left[Y_{j,i},Y_{j',0}\right]=\left[Y_{j,i},\frac{\partial}{\partial u^{j'}_{n+1}}\right]=0
\end{equation*}
since the coefficients of the vector-field $Y_{j,i}$ do not depend on the variables $u^j_{n+1}$ for any $j\in [1,N]$ due to explicit formula (\ref{vectorfieldsk}). Now consider general case assuming that $i=i'+k$, $k\geq 0$:
\begin{align*}
\left[Y_{j,i'+k},Y_{j',i'}\right]=\left[D_n^{-i'}Y_{j,k}D_n^{i'},D_n^{-i'}Y_{j',0}D_n^{i'}\right]=D_n^{-i'}\left[Y_{j,k},Y_{j',0}\right]D_n^{i'}.
\end{align*}
Latter vanishes due to the previous step.
\end{proof}

Explicit expressions for the commutators of  the basic operators $Y_{j,k}$ and $X_{j,k}$ with the operator $D_x$ of the total derivative with respect to $x$ provide an effective tool for studying the system (\ref{eqG}). It can easily be approved that the following operator identities hold
\begin{equation}\label{DxYj0}
\left[D_x,Y_{j,0}\right]=-\sum_{i=1}^N\sum_{k=1}^\infty Y_{j,0}\left(D_n^{k-1}F^i_n\right)\frac{\partial}{\partial u_{n+k}^i},
\end{equation}
\begin{equation}\label{DxYj1}
\left[D_x,Y_{j,1}\right]=-\sum_{i=1}^N\left(D_n^{-1}\left(Y_{j,0}\left(F^i_n\right)\right)Y_{i,1}+\sum_{k=1}^\infty Y_{j,1}\left(D_n^{k-1}F^i_n\right)\frac{\partial}{\partial u_{n+k}^i}+\sum_{k=1}^\infty Y_{j,1}\left(D_n^{1-k}G^i_n\right)\frac{\partial}{\partial u_{n-k}^i}\right),
\end{equation}
\begin{equation}\label{DxYjm}
\left[D_x,Y_{j,m}\right]=-\sum_{i=1}^N\left(\sum_{s=1}^m D_n^{-s}\left(Y_{j,m-s}\left(F^i_n\right)\right)Y_{i,s}+\sum_{k=1}^\infty Y_{j,m}\left(D_n^{k-1}F^i_n\right)\frac{\partial}{\partial u_{n+k}^i}+\sum_{k=1}^\infty Y_{j,m}\left(D_n^{1-k}G^i_n\right)\frac{\partial}{\partial u_{n-k}^i}\right),
\end{equation}
\begin{equation}\label{DxbarXj1}
\left[D_x,\bar X_{j,1}\right]=-\sum_{i=1}^N\sum_{k=1}^\infty \bar X_{j,1}\left(D_n^{1-k}G^i_n\right)\frac{\partial}{\partial u_{n-k}^i}.
\end{equation}
In a particular scalar case of the system (\ref{eqG}) when $N=1$ formulas (\ref{DxYj0})--(\ref{DxbarXj1}) have been derived in \cite{HPZ2009}.

Let $A_k$ be a ring of locally analytic functions of the variables $u_{n-k},u_{n-k+1},\ldots,u_n$; $u_{nx},u_{n,xx},u_{n,xxx},\ldots$. Here we formulate two theorems  relating integrability in the sense of Darboux and properties of the characteristic Lie-Rinehart algebras.

\begin{theorem} \label{Th1}
System (\ref{eqG}) admits a complete set of $n$-integrals if and only if the following two conditions hold:

1) The linear space  $V$ spanned by the operators $\{ Y_{i,s}\}$ has finite dimension, which we denote by $N_1$. Assume that $Z_1, Z_2,\ldots,Z_{N_1}$ constitute a basis in $V$, such that for any $Z\in V$ we have an expansion
\begin{equation*}
Z=\lambda_1Z_1+\lambda_2Z_2+\ldots+\lambda_{N_1}Z_{N_1},
\end{equation*}
We emphasize that the coefficients in this expansion, as a rule, are not constant, they are analytic functions of dynamic variables. Then due to construction of the space $V$ there exists a number $N_2$, such that $[Z_j,\bar X_{i,s}]=0$ for all $j=1,2,\ldots,N_1$, $i=1,2,\ldots,N$ and $s>N_2$.

2) The Lie-Rinehart algebra $L_n$ generated by the operators $\{Z_j\}_{j=1}^{N_1}$ and $\{\bar X_{i,s}\}_{s=1,i=1}^{N_2,N}$ over the ring $A_{N_2}$ has a finite dimension.
\end{theorem}

\begin{theorem}\label{Th2}
System (\ref{eqG}) admits a complete set of $x$-integrals if and only if the characteristic Lie-Rinehart algebra $L_x$ has a finite dimension.
\end{theorem}

The proof of these two theorems is beyond the scope of this paper. In the particular case when $N=1$ Theorems \ref{Th1}, \ref{Th2} are proved in \cite{HZS2010}. For systems of differential equations of hyperbolic type, similar statement is proved in \cite{Zhiber2007}. From these theorems we derive 

\begin{corollary}\label{Cor1}
System (\ref{eqG}) is integrable in the sense of Darboux if and only if both characteristic algebras $L_x$ and $L_n$ have finite dimension.
\end{corollary}

\section{Some properties of the characteristic algebras and their potential applications} 

In this section we will discuss some important properties of the characteristic algebras. Here we outline an approach to the problem of integrable classification of the lattices from the class (\ref{eq0}) based on algebraic ideas. Let us first concentrate on the algebra $L_x$ for the system (\ref{eqG}), generated as it was discussed in \S2 above by the operators $X_1,X_2,\ldots,X_N,K_0$. It is easy to see that the operator $D_x$ of the total derivative with respect to $x$ acts on the functions of the dynamical variables  $u_{n,x}$, $u_n$, $u_{n\pm1}$, $u_{n\pm2}$, $\dots$  as follows
\begin{align*}
D_xH(u_{n,x},u_n,u_{n\pm1}, u_{n\pm2},\ldots)=\left(\sum_{j=1}^{N}u^j_{n,xx}X_j+K_0\right)H(u_{n,x},u_n,u_{n\pm1}, u_{n\pm2},\ldots), 
\end{align*}
where the operators $K_0$ and $X_j$ have already been defined above (see (\ref{K0}), (\ref{Xk})). Since the operators $D_n$ and $D_x$ commute with each other we have a relation $D_nD_xD_n^{-1}=D_x$ or, the same
\begin{equation}\label{conjugation}
D_n\left(\sum_{j=1}^{N}u^j_{n,xx}X_j+K_0\right)D_n^{-1}=\sum_{j=1}^{N}u^j_{n,xx}X_j+K_0. 
\end{equation}
We simplify the left hand side of (\ref{conjugation}) due to the relations $D_n u^j_{n,xx}= u^j_{n+1,xx}D_n$ and due to
\begin{equation*}
u^j_{n+1,xx}=D_xF^j_n=\frac{\partial F^j_n}{\partial u^j_{n,x}}\cdot u^j_{n,xx}+K_0(F^j_{n}).
\end{equation*}
As a result we obtain the relation
\begin{equation}\label{simplification}
\sum_{j=1}^{N}\left(\frac{\partial F^j_n}{\partial u^j_{n,x}}\cdot u^j_{n,xx}+K_0\left(F^j_{n}\right)\right)D_n X_jD_n^{-1}+D_nK_0D_n^{-1}=\sum_{j=1}^{N}u^j_{n,xx}X_j+K_0. 
\end{equation}
Let us define an automorphism of the algebra $L_x$, acting according to the formula
\begin{equation}\label{automorphism}
Z\rightarrow D_nZ D_n^{-1}.
\end{equation}
It is remarkable that equation (\ref{simplification}) allows one to describe the action of the automorphism on the basic operators. Indeed, since variables $u^1_{n,xx}$, $u^2_{n,xx}$, \ldots, $u^N_{n,xx}$ are regarded as independent ones,  we can compare the coefficients in front of these variables in  (\ref{simplification}) and obtain the formulas
\begin{equation}\label{X}
D_nX_jD_n^{-1}=\frac{1}{ \frac{\partial F^j_n}{\partial u^j_{n,x}} }X_j,
\end{equation}
\begin{equation}\label{DKD}
D_nK_0D_n^{-1}=K_0 -\sum_{j=1}^{N} \frac{K_0(F^j_{n})}{ \frac{\partial F^j_n}{\partial u^j_{n,x}} }X_j.
\end{equation}

Elements of the algebra $L_x$ are vector fields with infinite number of components. Therefore, to prove relations of the form $Y=0$ for $Y\in L_x$, it is necessary to check an infinite set of conditions. The following lemma provides a convenient tool for exploring such questions.

\begin{lemma} \label{Lem1} Assume that the vector field
\begin{equation*}
\label{Klemma}
K=\sum_{j=1}^{N}\sum_{k=1}^{\infty}\left(\alpha^j(k)\frac{\partial}{\partial u^j_{n+k}}+\alpha^j(-k)\frac{\partial}{\partial u^j_{n-k}}\right),
\end{equation*}
solves the equation
\begin{equation}\label{DKlemma}
D_nKD^{-1}_n=hK,
\end{equation}
where the factor $h$ is a function of the dynamical variables, then $K=0$.
\end{lemma}

\begin{proof} First, we find an explicit representation of the operator $D_nKD_n^{-1}$ and then by substituting the result into equation (\ref{DKlemma}) we obtain 
\begin{align} \label{eqlemma}
&D_n(\alpha^j(-1))\sum^N_{s=1}D_n\left(\frac{\partial}{\partial u^j_{n-1}}G^s_n\right)X_s+ D_n(\alpha^j(-1))\frac{\partial}{\partial u_n^j}+D_n(\alpha^j(-2))\frac{\partial}{\partial u^j_{n-1}}\nonumber\\
&+\sum_{k=2}^{\infty}D_n(\alpha^j(k-1))\frac{\partial}{\partial u^j_{n+k}}+D_n(\alpha^j(-k-1))\frac{\partial}{\partial u^j_{n-k}}\\
&=h\left(\sum_{k=1}^{\infty}\alpha^j(k)\frac{\partial}{\partial u^j_{n+k}}+\alpha^j(-k)\frac{\partial}{\partial u^j_{n-k}}\right) \nonumber
\end{align}
for any $j$ from the segment $1\leq j\leq N$.

By comparing the coefficients before the operators $\frac{\partial}{\partial u^j_{n+k}}$ and $\frac{\partial}{\partial u^j_{n-k}}$ in the equation (\ref{eqlemma}) one can easily prove that $\alpha^j(k)=0$ for all values of $k$.
\end{proof}

For the case of the algebra $L_n$ we have a similar statement.

\begin{lemma} \label{Lemma2} Suppose that the vector field
\begin{equation}\label{Ylemma} 
Y=\sum_{j=1}^{N}\left(\alpha^j(1)\frac{\partial}{\partial u^j_{n,x}}+\alpha^j(2)\frac{\partial}{\partial u^j_{n,xx}}+ \alpha^j(3)\frac{\partial}{\partial u^j_{n,xxx}}+\cdots\right)
\end{equation}
solves the equation
\begin{equation}\label{DxYlemma} 
[D_x,Y]=hY,
\end{equation}
where $h$ is a function of the dynamical variables, then $Y=0$.
\end{lemma} 

\begin{proof}
From (\ref{Ylemma}) and (\ref{DxYlemma}) we easily obtain the equation 
\begin{align*} \label{eqlemma2}
&\sum_{j=1}^{N}\left(-\alpha^j(1)\frac{\partial F^j_n}{\partial u^j_{n,x}}  \frac{\partial}{\partial u^j_{n+1}}-\alpha^j(1)\frac{\partial}{\partial u^j_n}+(D_x\alpha^j(1)-\alpha^j(2)) \frac{\partial}{\partial u^j_{n,x}} + (D_x\alpha^j(2)-\alpha^j(3)) \frac{\partial}{\partial u^j_{n,xx}}+\cdots\right)\nonumber\\
&=h\sum_{j=1}^{N}\left(\alpha^j(1)\frac{\partial}{\partial u^j_{n,x}}+\alpha^j(2)\frac{\partial}{\partial u^j_{n,xx}}+ \alpha^j(3)\frac{\partial}{\partial u^j_{n,xxx}}+\cdots \right).    
\end{align*}
Now we compare the coefficients before the operators $\frac{\partial}{\partial u^j_{n+1}}$, $\frac{\partial}{\partial u^j_n}$, $\frac{\partial}{\partial u^j_{n,x}}$, $\frac{\partial}{\partial u^j_{n,xx}}$, \ldots \quad to make sure that $\alpha^j(k)=0$ for~all $k$. That completes the proof of Lemma~\ref{Lemma2}.
\end{proof}

\section{Characteristic algebras for the reduced system (\ref{eq00})}

In this section we concentrate on the differential-difference system (\ref{eq00}) being a reduction of the differential-difference equation (\ref{eq0}). System (\ref{eq00}) is a particular case of the general system (\ref{eqG}), where assumed that 
\begin{align*}
&F^1=H^{(1)}(u_{n,x}^1,u_n^{2},u_{n+1}^1,u_n^1), \\
&F^j=F(u_{n,x}^j,u_n^{j+1},u_{n+1}^j,u_n^j,u_{n+1}^{j-1}) \quad \mbox{for} \quad 2 \leq j \leq N-1,\\
&F^N=H^{(2)}(u_{n,x}^N,u_{n+1}^N,u_n^N,u_{n+1}^{N-1}). 
\end{align*} 
Suppose that system (\ref{eq00}) is integrable, then its characteristic algebra $L_x$ has a finite dimension. The problem is figuring out what this means in terms of the right hand side of the differential-difference equation (\ref{eq0}), i.e. in terms of the function $F(u_{n,x}^j,u_n^{j+1},u_{n+1}^j,u_n^j,u_{n+1}^{j-1})$. Below we make the first step towards solving the problem. We define a sequence of the operators in $L_x$ according to the rule
\begin{equation}\label{comKj}
K_1=[X_j,K_0],\quad K_2=[X_j,K_1], \quad \ldots,\quad K_m=[X_j,K_{m-1}],\quad \ldots.
\end{equation}
Since the algebra $L_x$ is finite dimensional then there exists an integer $M$ such that operator $K_M$ is represented as a linear combination of the preceding members of the sequence:
\begin{equation}\label{KM}
K_M=\lambda K_{M-1}+\lambda_{1}K_{M-2}+ \dots +\lambda_{M-1}K_1,
\end{equation}
where the operators $K_{M-1}, K_{M-2}, \ldots, K_{1}$ are linearly independent. 

Explicit formula (\ref{K0}) shows that operator $K_1=[X_j,K_0]$ has the following coordinate representation 
\begin{equation*}\label{K1}
K_1=\frac{\partial}{\partial u_{n}^j}+X_j(F_{n}^j)\frac{\partial}{\partial u_{n+1}^j}+ X_j(G_{n}^j) \frac{\partial}{\partial u^j_{n-1}} +X_j(F_{n+1}^j)\frac{\partial}{\partial u_{n+2}^j}+ X_j(G_{n-1}^j) \frac{\partial}{\partial u^j_{n-2}}+\ldots,
\end{equation*}
while for $m\geq 2$ we have
\begin{equation}\label{Km}
K_m=X^m_j(F_{n}^j)\frac{\partial}{\partial u_{n+1}^j}+ X^m_j(G_{n}^j) \frac{\partial}{\partial u^j_{n-1}} +X^m_j(F_{n+1}^j)\frac{\partial}{\partial u_{n+2}^j}+ X^m_j(G_{n-1}^j) \frac{\partial}{\partial u^j_{n-2}}+\ldots.
\end{equation}
Since operator $K_1$ contains the term $\frac{\partial}{\partial u_{n}^j}$ and the other operators in (\ref{KM}) do not, then the coefficient $\lambda_{M-1}$ vanishes.

Automorphism (\ref{automorphism}) provides an effective tool for studying the problem of decomposition (\ref{KM}). Above in (\ref{X}), (\ref{DKD}) we have already described the action of the automorphism on the basic operators. We can derive similar formulas for the members of the sequence (\ref{comKj})
\begin{align}
&D_nK_1D^{-1}_n=q_jK_1-q_j^2K_1(F^j_n)X_j-q_{j-1}F^{j-1}_{n,u_n^j}q_jX_{j-1}-q_{j+1}F^{j+1}_{n,u_n^j}X_{j+1},\label{DK1D}\\
&D_nK_2D^{-1}_n=q_j^2K_2+q_jX_j(q_j)K_1+r_2X_j-q_{j-1}F^{j-1}_{u_n^j}q_jX_j(q_j)X_{j-1},\label{DK2D}\\
&D_nK_3D^{-1}_n=q_j^3K_3+3q_j^2X(q_j)K_2+(q_jX_j)^2(q_j)K_1+r_3X_j-q_{j-1}F^{j-1}_{n,u_n^j}\left\{(q_jX_j)^2(q_j)\right\}X_{j-1}.\label{DK3D}
\end{align}
Here $q_j=\left(\frac{\partial F^j_n}{\partial u^j_{n,x}}\right)^{-1}$, and $r_2$, $r_3$ are functions on $q_j$ and its derivatives. For the general case we can prove by induction that 
\begin{align}\label{DKmD}
\begin{aligned}
D_nK_mD^{-1}_n=&q_j^mK_m+p_mK_{m-1}+s_mK_{m-2}+\ldots\\
&+(q_jX_j)^{m-1}(q_j)K_1+r_mX_j+q_{j-1}F^{j-1}_{n,u_n^j}\left\{(q_jX_j)^{m-1}(q_j)\right\}X_{j-1},
\end{aligned}
\end{align}
where 
\begin{align*}
&p_m=\frac{1}{2}m(m-1)q_j^{m-1}X_j(q_j),\\
&s_m=\frac{1}{6}m(m-1)(m-2)q_j^{m-1}X_j^2(q_j)+\frac{1}{24}m(m-1)(m-2)(3m-5)q^{m-2}(X_j(q_j))^2.
\end{align*}
Explicit expression for $r_m$ is rather complicated and we do not specify it.

Let us apply automorphism (\ref{automorphism}) to both sides of the equation (\ref{Km}) and get
\begin{equation}\label{equality}
\begin{aligned}
q^M_jK_M&+p_MK_{M-1}+s_MK_{M-2}+\ldots\\
&+(q_jX_j)^{M-1}(q_j)K_1+r_MX_j+q_{j-1}F^{j-1}_{n,u_n^j}\left\{(q_jX_j)^{M-1}(q_j)\right\}X_{j-1}\\
&=D_n(\lambda)\left(q_j^{M-1}K_{M-1}+\ldots\right)+D_n(\lambda_1)\left(q_j^{M-2}+\ldots\right)+\ldots.
\end{aligned}
\end{equation}
We first replace $K_M$ due to (\ref{Km}) and then collect in (\ref{equality}) the coefficients in front of the linearly independent operators $K_{M-1}$, $K_{M-2}$, \ldots, $K_1$: 
\begin{align}
&K_{M-1}: \quad q_j^M\lambda+p_M=D_n(\lambda)q_j^{M-1}, \label{eq1_lam}\\
&K_{M-2}: \quad q_j^M\lambda_1+s_M=D_n(\lambda)p_{M-1}+D_n(\lambda_1)q_j^{M-2}, \label{eq2_lam}\\
&\ldots\ldots\ldots\ldots\ldots \nonumber\\
&K_{1}:   \qquad (q_jX_j)^{M-1}q_j=D_n(\lambda)(q_jX_j)^{M-2}q_j+\ldots+D_n(\lambda)q_jX_j(q_j). \label{eq3_lam}
\end{align}
Due to the explicit expression (\ref{Km}) the coefficients of the operators $K_2$, $K_3$, \ldots, $K_M$ depend on $u_{n,x}^j$ and on the variables $u^{j+1}_{n}$, $u^{j}_{n}$, $u^{j-1}_{n}$ and their shifts on $n$. The factors $\lambda,\lambda_1,\ldots,\lambda_{M-2}$ might depend only on these variables. However we have relations (\ref{eq1_lam})--(\ref{eq3_lam}) which provide an additional restriction for the factors. Indeed the coefficients $q_j$, $p_m$, $s_m$ in (\ref{eq1_lam})--(\ref{eq3_lam}) depend only on the variables $u^j_{n,x}$, $u^{j+1}_{n}$, $u^{j}_{n+1}$, $u^{j}_{n}$, $u^{j-1}_{n+1}$, hence $\lambda,\lambda_1,\ldots,\lambda_{M-2}$ might depend only on the variables $u_{n,x}^j$, $u^j_n$.

Relations (\ref{eq1_lam})--(\ref{eq3_lam}) determine a system of functional equations with the set of unknown functions $\lambda$, $\lambda_1$, \ldots, $\lambda_{M-2}$ depending on two variables $u_{n,x}^j$, $u^j_n$. Indeed, each of these equations contains the unknowns taken at different points $(u_{n,x}^j, u^j_n)$ and $\left(D_nu^j_{n,x}, D_nu^j_n\right)=\left(F\left(u^j_{n,x},u^{j+1}_{n},u^{j}_{n+1},u^{j}_{n},u^{j-1}_{n+1}\right),u^j_{n+1}\right)$. An important peculiarity of the system (\ref{eq1_lam})--(\ref{eq3_lam}) is that it is highly overdetermined since the coefficients of the equations depend on five independent variables while the solutions depend only on two of them. 

Let us finalize the reasoning above as a statement.

\begin{theorem} {\bf (Necessary condition of integrability of system (\ref{eq00}) and hence due to Conjecture~\ref{conjecture} of the equation (\ref{eq0}).)}
If system (\ref{eq00}) admits the complete set of $x$-integrals then the overdetermined system of equations (\ref{eq1_lam})--(\ref{eq3_lam}) admits a solution $\left(\lambda,\lambda_1,\ldots,\lambda_{M-2}\right)$ depending only on $u_{n,x}^j$, $u^j_n$ for any $j$ from the segment $1\leq j\leq N$.
\end{theorem}

Below we show how to reduce the system (\ref{eq1_lam})--(\ref{eq3_lam}) to a system of differential equations by using the characteristic operators in the direction of $n$. We start with the equation (\ref{eq1_lam}) that is specified to the form 
\begin{align*}
\lambda+\varepsilon\frac{1}{q_j}X_j(q_j)=D_n(\lambda)\frac{1}{q_j},
\end{align*}
where $\varepsilon=\frac{M(M-1)}{2}$. Since $q_j=\left(\frac{\partial F^j_n}{\partial u^j_{n,x}}\right)^{-1}$ equation reduces to 
\begin{equation}\label{qF}
\lambda \frac{\partial F^j_n}{\partial u^j_{n,x}}-\varepsilon\frac{\partial^2 F^j_n}{\partial (u^j_{n,x})^2}=D_n(\lambda)\left(\frac{\partial F^j_n}{\partial u^j_{n,x}}\right)^2.
\end{equation} 
We can rewrite the latter equation as 
\begin{equation}\label{eq_lambda2}
\lambda-\varepsilon\frac{\partial}{\partial u^j_{n,x}}\log \frac{\partial F^j_n}{\partial u^j_{n,x}}=D_n(\lambda)\frac{\partial F^j_n}{\partial u^j_{n,x}}.
\end{equation}
Let us apply the operator $\frac{\partial}{\partial u^j_{n+1}}$ to both sides of the equation (\ref{eq_lambda2}). By taking into account that $\lambda=\lambda\left(u^j_{n,x},u^j_{n}\right)$ we obtain 
\begin{equation}\label{eq_lambda3}
-\varepsilon\frac{\partial}{\partial u^j_{n+1}}\frac{\partial}{\partial u^j_{n,x}}\left(\log \frac{\partial F^j_n}{\partial u^j_{n,x}}\right)=\frac{\partial}{\partial u^j_{n+1}}\left(D_n(\lambda)\frac{\partial F^j_n}{\partial u^j_{n,x}}\right).
\end{equation}
As it was discussed above we have the relation $D_n^{-1}\frac{\partial}{\partial u^j_{n+1}}D_n=Y_{j,1}$ (see formula (\ref{vectorfields})) when the operators act on the variables $u_n^j$, $u_{n-1}^j$, $u_{n,x}^j$. Since $\frac{\partial}{\partial u^j_{n,x}}=X_j$ we can represent (\ref{eq_lambda3}) in the form 
\begin{equation}\label{eq_lambda4}
-\varepsilon Y_{j,1}\left\{D_n^{-1}X_j\log X_j(F^j_n)\right\}=Y_{j,1}\left\{\lambda D_n^{-1}(X_j(F^j_n))\right\}.
\end{equation}
In what follows we need some rather simple formulas which are produced by the obvious identity 
\begin{equation*}
u^j_{n,x}=G^j\left(F^j\left(u^j_{n,x},u^{j+1}_{n},u^{j}_{n+1},u^{j}_{n},u^{j-1}_{n+1}\right),u^{j+1}_{n},u^{j}_{n+1},u^{j}_{n},u^{j-1}_{n+1}\right).
\end{equation*}
Differentiation of the identity with respect to the variable $u^j_{n,x}$ gives rise to equation 
\begin{equation*}
1=\frac{\partial D_nG^j_n}{\partial F^j_n}\times\frac{\partial F^j_n}{\partial u^j_{n,x}}
\end{equation*}
which implies 
\begin{equation*}
1=\left(D_n\frac{\partial G^j_n}{\partial u^j_{n,x}}\right)\left(\frac{\partial F^j_n}{\partial u^j_{n,x}}\right)
\end{equation*}
or the same 
\begin{equation}\label{relation_1}
D_n^{-1}X_jF^j_n=\frac{1}{X_jG^j_n}.
\end{equation}
By combining (\ref{X}) and (\ref{relation_1}) we find 
\begin{equation}\label{relation_2}
D_n^{-1}X_jD_n=\frac{1}{X_jG^j_n}X_j.
\end{equation}
By differentiating identity with respect to $u^j_{n+1}$ we obtain an equation 
\begin{equation*}
0=D_n\left(\frac{\partial G^j_n}{\partial u^j_{n}}\right)+D_n\left(\frac{\partial G^j_n}{\partial u^j_{n,x}}\right)\frac{\partial F^j_n}{\partial u^j_{n+1}},
\end{equation*}
which implies 
\begin{align*}\label{dif_Y_h}
D_n^{-1}\frac{\partial F^j_n}{\partial u^j_{n+1}}=-\frac{\partial G^j_n/\partial u_n^j}{\partial G^j_n/\partial u^j_{n,x}}.
\end{align*}
Now equation (\ref{eq_lambda4}) is easily converted due to (\ref{relation_1}), (\ref{relation_2}) to the form 
\begin{equation}\label{Y_lambda}
\varepsilon Y_{j,1}\left(\frac{X_j^2G^j_n}{(X_jG^j_n)^2}\right)=Y_{j,1}\left(\frac{\lambda}{X_jG^j_n}\right).
\end{equation}

Let as replace the characteristic operator $Y_{j,1}$ by its explicit representation (\ref{vectorfields}). Since in (\ref{Y_lambda}) the operator is applied to functions depending on the variables $u^j_{n,x}$, $u^{j+1}_{n-1}$, $u^{j}_{n}$, $u^{j}_{n-1}$, $u^{j-1}_{n}$ only we omit the terms which we do not use. More precisely we use the relation that holds for the function h depending on these variables 
\begin{align*}\label{Y_h}
Y_{j,1}h=\left(\frac{\partial}{\partial u^j_{n}}-D_n^{-1}\left(\frac{\partial F^j_n}{\partial u^j_{n+1}}\right)\frac{\partial}{\partial u^j_{n,x}}\right)h.
\end{align*}
Therefore equation (\ref{Y_lambda}) gets the form 
\begin{equation}\label{final_eq_lambda}
\left(\frac{\partial}{\partial u^j_{n}}-\frac{\partial G^j_n/\partial u_n^j}{\partial G^j_n/\partial u^j_{n,x}}\frac{\partial}{\partial u^j_{n,x}}\right)\left(\frac{\lambda}{\partial G^j_n/\partial u^j_{n,x}}-\varepsilon \frac{\partial^2 G^j_n/\partial (u_{n,x}^j)^2}{\left(\partial G^j_n/\partial u^j_{n,x}\right)^2}\right)=0.
\end{equation}

Thus the functional equation (\ref{eq1_lam}) produces a differential consequence (\ref{final_eq_lambda}). In a similar way we can derive differential consequences for the other functional equations of the system (\ref{eq1_lam})--(\ref{eq3_lam}). For example (\ref{eq2_lam}) implies 
\begin{equation*}\label{final_eq2_lambda}
\left(\frac{\partial}{\partial u^j_{n}}-\frac{\partial G^j_n/\partial u_n^j}{\partial G^j_n/\partial u^j_{n,x}}\frac{\partial}{\partial u^j_{n,x}}\right)\left(\delta\frac{X_j^3(G^j_n)}{(X_j(G^j_n))^3}+\gamma\frac{(X_j^2(G^j_n))^2}{(X_j(G^j_n))^4}-\varepsilon_1\frac{\lambda X_j^2(G^j_n)}{(X_j(G^j_n))^2}+\frac{\lambda_1}{(X_j(G^j_n))^2}\right)=0,
\end{equation*}
where $\varepsilon_1=\frac{(M-1)(M-2)}{2}$, $\delta=\frac{M(M-1)(M-2)}{6}$, $\gamma=\frac{M(M-1)(M-2)(M-3)}{8}$.

By using the identity $(q_jX_j)^mq_j=X^{m+1}_j(G^j_n)$, that is easily proved by induction one can rewrite the last equation of the system as follows 
\begin{equation*}\label{final_eq3_lambda}
\left(X_j^M-\lambda X_j^{M-1}-\lambda_1 X_j^{M-2}- \ldots -\lambda_{M-2} X_j^2\right)G^j_n=0.
\end{equation*}

Finally we can state that the overdetermined system of functional equations (\ref{eq1_lam})--(\ref{eq3_lam}) is reduced to an overdetermined system of differential equations, that is certainly more simple. An equation (\ref{final_eq_lambda}) is of a special interest, it should admit a solution $\lambda=\lambda\left(u_n^j,u_{n,x}^j\right)$ depending only on two variables $u_n^j$ and $u_{n,x}^j$. It is a severe requirement since the coefficients of the equation (\ref{final_eq_lambda}) depend generally on five variables $u^j_{n,x}$, $u^{j+1}_{n-1}$, $u^{j}_{n}$, $u^{j}_{n-1}$, $u^{j-1}_{n}$.

\begin{example}
We consider an illustrative example by taking a scalar equation of the form
\begin{equation}\label{example}
u_{n+1,x}=u_{n,x}u_{n+1}^2.
\end{equation} 
Here $F_n=u_{n,x}u_{n+1}^2$, $\frac{\partial F_n}{\partial u_{n,x}}=u_{n+1}^2$ and $\frac{\partial^2 F_n}{\partial (u_{n,x})^2}=0$. Equation (\ref{qF}) turns into 
\begin{equation}\label{exampleequation}
\frac{D_n(\lambda)}{\lambda}=\frac{1}{u_{n+1}^2}
\end{equation}
Due to equation (\ref{example}) we can rewrite (\ref{exampleequation}) as 
\begin{equation*}
\frac{D_n(\lambda)}{\lambda}=\frac{u_{n,x}}{u_{n+1,x}}.
\end{equation*}
Evidently it has a solution of the necessary form $\lambda=\frac{c}{u_{n,x}}$ where $c$ is a constant. 
\end{example}

\section{Examples, approving the Conjecture~\ref{conjecture}}

The aim of this section is discussing the integrable lattices of the form (\ref{eq0}) given in \cite{FNR}. Since these equations are written as
\begin{equation*}\label{eq0fer}  
v_{k+1,x}^s =\bar F(v_{k,x}^{s+1},v_k^{s},v_{k+1}^s,v_k^{s+1},v_{k+1}^{s+1}), \quad -\infty < s,k < \infty,
\end{equation*}
at first we convert them into the form (\ref{eq0}) by a linear transformation of the independent variables: $v_k^s=u_n^{j+1}$, $n=k$, $j=-k-s$. Below we give the list of integrable differential-difference equations from \cite{FNR} rewritten in the new variables
\begin{align}
& 1) \quad u^{j}_{n+1,x}=u^{j}_{n,x}+e^{u^{j}_{n+1}-u^{j+1}_{n}}-e^{u^{j-1}_{n+1}-u^{j}_{n}};\label{eq1}\\
& 2) \quad u^{j}_{n+1,x}=u^{j}_{n,x}+\frac{e^{u^{j-1}_{n+1}}}{e^{u^{j}_{n+1}}}-\frac{e^{u^{j-1}_{n+1}}}{e^{u^{j}_{n}}}-\frac{e^{u^{j}_{n}}}{e^{u^{j+1}_{n}}}+\frac{e^{u^{j}_{n+1}}}{e^{u^{j+1}_{n}}};\label{eq2}\\
& 3) \quad u^{j}_{n+1,x}=u^{j}_{n,x}\frac{\left(u^{j}_{n+1}\right)^2}{u^{j-1}_{n+1}u^{j+1}_{n}};\label{eq3}\\
& 4) \quad u^{j}_{n+1,x}=u^{j}_{n,x}\frac{\left(u^{j}_{n+1}-u^{j+1}_{n}\right)}{\left(u^{j-1}_{n+1}-u^{j}_{n}\right)};\label{eq4}\\
& 5) \quad u^{j}_{n+1,x}=u^{j}_{n,x}\frac{u^{j}_{n+1}\left(u^{j}_{n+1}-u^{j+1}_{n}\right)}{u^{j+1}_{n}\left(u^{j-1}_{n+1}-u^{j}_{n}\right)};\label{eq5}\\
& 6) \quad u^{j}_{n+1,x}=u^{j}_{n,x}\frac{\left(u^{j-1}_{n+1}-u^{j}_{n+1}\right)\left(u^{j}_{n+1}-u^{j+1}_{n}\right)}{\left(u^{j-1}_{n+1}-u^{j}_{n}\right)\left(u^{j}_{n}-u^{j+1}_{n}\right)};\label{eq6}\\
& 7) \quad u^{j}_{n+1,x}=u^{j}_{n,x}\frac{\sinh\left(u^{j-1}_{n+1}-u^{j}_{n+1}\right)\sinh\left(u^{j}_{n+1}-u^{j+1}_{n}\right)}{\sinh\left(u^{j-1}_{n+1}-u^{j}_{n}\right)\sinh\left(u^{j}_{n}-u^{j+1}_{n}\right)}\label{eq7}
\end{align}
their Lax pairs, also found in \cite{FNR} are respectively, of the form
\begin{align} 
& 1) \quad \begin{cases}\label{Laxeq1}
\varphi^{j}_{n+1}=-\varphi^{j+1}_{n}+e^{u^{j}_{n}-u^{j}_{n+1}}\varphi^{j}_{n},\\
\varphi^{j}_{n,x}=-e^{u^{j-1}_{n}-u^{j}_{n}}\varphi^{j-1}_{n};
\end{cases} \\
& 2) \quad \begin{cases}\label{Laxeq2}
\varphi^{j}_{n+1}=-e^{-u^{j}_{n+1}+u^{j}_{n}}\left(\varphi^{j+1}_{n}-\varphi^{j}_{n}\right)+\varphi^{j+1}_{n},\\
\varphi^{j}_{n,x}=e^{u^{j-1}_{n}-u^{j}_{n}}\left(\varphi^{j-1}_{n}-\varphi^{j}_{n}\right);
\end{cases}\\
& 3) \quad \begin{cases}\label{Laxeq3}
\varphi^{j}_{n+1}=-\frac{u^{j}_{n+1}}{u^{j+1}_{n}}\left(\varphi^{j}_{n}-\varphi^{j+1}_{n}\right),\\
\varphi^{j}_{n,x}=-\frac{u^{j}_{n,x}}{u^{j}_{n}}\left(\varphi^{j-1}_{n}-\varphi^{j}_{n}\right);
\end{cases} \\
& 4) \quad \begin{cases}\label{Laxeq4}
\varphi^{j}_{n+1}=\left(u^{j+1}_{n}-u^{j}_{n+1}\right)\varphi^{j}_{n}+\varphi^{j+1}_{n},\\
\varphi^{j}_{n,x}=-u^{j}_{n,x}\varphi^{j-1}_{n};
\end{cases} \\
& 5) \quad \begin{cases}\label{Laxeq5}
\varphi^{j}_{n+1}=\left(1-\frac{u^{j}_{n+1}}{u^{j+1}_{n}}\right)\varphi^{j}_{n}-\frac{u^{j}_{n+1}}{u^{j+1}_{n}}\varphi^{j+1}_{n},\\
\varphi^{j}_{n,x}=\frac{u^{j}_{n,x}}{u^{j}_{n}}\left(\varphi^{j-1}_{n}+\varphi^{j}_{n}\right);
\end{cases}\\
& 6) \quad \begin{cases}\label{Laxeq6}
\varphi^{j}_{n+1}=\frac{u^{j}_{n+1}-u^{j+1}_{n}}{u^{j}_{n}-u^{j+1}_{n}}\varphi^{j}_{n}+\left(1-\frac{u^{j}_{n+1}-u^{j+1}_{n}}{u^{j}_{n}-u^{j+1}_{n}}\right)\varphi^{j+1}_{n},\\
\varphi^{j}_{n,x}=\frac{u^{j}_{n,x}}{u^{j-1}_{n}-u^{j}_{n}}\left(\varphi^{j-1}_{n}+\varphi^{j}_{n}\right);
\end{cases}\\
& 7) \quad \begin{cases}\label{Laxeq7}
\varphi^{j}_{n+1}=\frac{e^{2\left(u^{j}_{n+1}-u^{j+1}{n}\right)}-1}{e^{2\left(u^{j}_{n}-u^{j+1}_{n}\right)}-1}\varphi^{j}_{n}+\left(1-\frac{e^{2\left(u^{j}_{n+1}-u^{j+1}_{n}\right)}-1}{e^{2\left(u^{j}_{n}-u^{j+1}_{n}\right)}-1}\right)\varphi^{j+1}_{n},\\
\varphi^{j}_{n,x}=\frac{2u^{j}_{n,x}}{e^{2\left(u^{j-1}_{n}-u^{j}_{n}\right)}-1}\left(\varphi^{j-1}_{n}-\varphi^{j}_{n}\right).
\end{cases}
\end{align}

It can be shown that equations (\ref{eq1})-(\ref{eq7}) admit infinite sequences of suitable finite field reductions being integrable in the sense of Darboux systems of differential-difference hyperbolic type equations. Simultaneously, the Lax pairs (\ref{Laxeq1})-(\ref{Laxeq7}) of equations also pass into the Lax pairs of the corresponding reductions. The crucial point of our algorithm is finding a special constraint consistent with the lattice that divides the lattice into two independent parts. We call such kind of constraint a degenerate boundary condition. Below we look for the degenerate boundary conditions for the equations (\ref{eq1})-(\ref{eq7}) and appropriate boundary conditions for their Lax pairs. 

The Lax pair is an important attribute of the integrable soliton equations. They also have useful applications for Darboux integrable systems. For example, in some cases, Lax pairs provide an effective tool for constructing characteristic integrals. We will now outline a simple algorithm suitable for this purpose.

Assume that system of equations (\ref{eqG}) admits a Lax pair, i.e. it is a compatibility condition of a pair of systems of linear equations 
\begin{align}
&\Phi_{n+1}=U_n\Phi_n\label{n},\\
&\Phi_{n,x}=V_n\Phi_n\label{x},
\end{align}
where the potentials $U_n$ and $V_n$ of the linear systems are matrices with the following triangular structure
\begin{equation}\label{U}
U_n=\left(\begin{array}{cccc}
a_{11,n}&a_{12,n}&\dots &a_{1N,n} \\
0&a_{22,n}&\dots &a_{2N,n} \\
\dots &\dots&\dots& \dots \\
0&0&\dots &a_{NN,n}
\end{array} \right),
\end{equation}
\begin{equation}\label{V}
V_n=\left(\begin{array}{cccc}
b_{11,n}&0&\dots &0 \\
b_{21,n}&b_{22,n}&\dots &0 \\
\dots &\dots&\dots& \dots \\
b_{N1,n}&b_{N2,n}&\dots &b_{NN,n}
\end{array} \right).
\end{equation}
Let $P^{(k)}_n$ be a product of the matrices
\begin{equation*}\label{P}
P^{(k)}_n=U_{n+k} U_{n+k-1}\cdots U_{n},
\end{equation*}
then obviously we have
\begin{equation}\label{DxP}
D_xP^{(k)}_n=V_{n+k+1}P^{(k)}_n-P^{(k)}_nV_{n}.
\end{equation}
Evidently $P^{(k)}_n$ is an upper triangular matrix
\begin{equation*} \label{P_matrix}
P^{(k)}_n=\left(\begin{array}{cccc}
p^{(k)}_{11,n}&p^{(k)}_{12,n}&\dots &p^{(k)}_{1N,n} \\
0&p^{(k)}_{22,n}&\dots &p^{(k)}_{2N,n} \\
\dots &\dots&\dots& \dots \\
0&0&\dots &p^{(k)}_{NN,n}
\end{array} \right).
\end{equation*}
In the equality (\ref{DxP}), we select the matrix elements located at the intersection of the first row and the last column. As a result, we obtain a scalar equality of the form
\begin{equation}\label{DxP1N}
D_xp^{(k)}_{1N,n}=p^{(k)}_{1N,n}(D_n^{k+1}b_{11}-b_{NN}).
\end{equation}
Now we can conclude:

\begin{lemma} \label{LemP-x}
If the entries of the matrix (\ref{V}) satisfy the relations $D_nb_{11}=b_{11}$ and $b_{11}=b_{NN}$ then function $J=p^{(k)}_{1N,n}$ is an $x$-integral.
\end{lemma} 

\begin{proof} 
In such a case (\ref{DxP1N}) implies $D_xp^{(k)}_{1N,n}=0$.
\end{proof}

For constructing $n$-integrals we use the higher order derivatives of the equation system (\ref{x}) with respect to $x$. Evidently they solve linear systems as
\begin{align*}
&\Phi_{n,xx}=(V_{n,x}+V_n^2)\Phi_n, \\
&\Phi_{n,xxx}=(V_{n,xx}+2V_nV_{n,x}+V_{n}V_{n,x}+V_n^3))\Phi_n.
\end{align*}
For arbitrary $k$ we find
\begin{align*}
D_x^k(\Phi_{n})=R^{(k)}_n\left(V_{n},V_{n,x},\ldots,D_x^{(k-1)}(V_{n})\right)\Phi_n,
\end{align*}
where $R^{(k)}_n$ is a polynomial with constant coefficients on all of its arguments 
\begin{equation*} \label{R}
R^{(k)}_n=\left(\begin{array}{cccc}
r^{(k)}_{11,n}&0&\dots &0 \\
r^{(k)}_{21,n}&r^{(k)}_{22,n}&\dots &0 \\
\dots &\dots&\dots& \dots \\
r^{(k)}_{N1,n}&r^{(k)}_{N2,n}&\dots &r^{(k)}_{NN,n}
\end{array} \right).
\end{equation*}

By applying the shift operator $D_n$ to the latter equation we obtain
\begin{equation}\label{Dxk1}
D_x^k\Phi_{n+1}=D_n(R^{(k)}_n)U_n\Phi_n.
\end{equation}
On the other hand by differentiating (\ref{n}) with respect to $x$ $k$ times we find
\begin{equation}\label{Dxk2}
D_x^k\Phi_{n+1}=\sum_{j=0}^{k}c^j_kD_x^j(U_n)  (R^{(k-j)}_n)\Phi_n.
\end{equation}
Equations (\ref{Dxk1}) and (\ref{Dxk2}) evidently imply
\begin{equation}\label{Dxk12}
D_n(R^{(k)}_n)U_n=\sum_{j=0}^{k}c^j_kD_x^j(U_n)  (R^{(k-j)}_n).
\end{equation}
Let us  pass in the equality (\ref{Dxk12}),  to the matrix elements located at the left lower corner. As a result we arrive at the equation
\begin{equation}\label{DnN1}
D_n(r^{(k)}_{N1,n})a_{11}=a_{NN}r^{(k)}_{N1,n}+a_{NN,x}r^{(k-1)}_{N1,n}+a_{NN,xx}r^{(k-2)}_{N1,n}+\cdots +D^{(k)}_x(a_{NN}).
\end{equation}

\begin{lemma} \label{LemRk-n}
If the entries of the matrix (\ref{U}) satisfy the conditions $D_xa_{11}=0$ and $a_{11}=a_{NN}$ then function $I=r^{(k)}_{N1,n}$ is an $n$-integral.
\end{lemma}

\begin{proof} 
Proof follows right away from equation  (\ref{DnN1}).
\end{proof}

\subsection{Reductions of the equation (\ref{eq1})}
To find the desired degenerate boundary condition for (\ref{eq1}) we make the following change of the variables:
\begin{align*}
&u_n^j=v_n^j-\log \varepsilon\quad \mbox{for}\quad j>0,\quad \varepsilon>0,\\
&u_n^0=v_n^0,\\
&u_n^j=v_n^j+\log \varepsilon\quad \mbox{for}\quad j<0. 
\end{align*}
Then (\ref{eq1}) converts into
\begin{align*}
&v_{n+1,x}^0=v_{n,x}^0+\varepsilon e^{v^{0}_{n+1}-v^{1}_{v}}-\varepsilon e^{v^{-1}_{n+1}-v^{0}_{n}},\\
&v^{1}_{n+1,x}=v^{1}_{n,x}+e^{v^{1}_{n+1}-v^{2}_{n}}-\varepsilon e^{v^{0}_{n+1}-v^{1}_{n}},\\
&v^{-1}_{n+1,x}=v^{-1}_{n,x}+\varepsilon e^{v^{-1}_{n+1}-v^{0}_{n}}- e^{v^{-2}_{n+1}-v^{-1}_{n}},\\
&v^{j}_{n+1,x}=v^{j}_{n,x}+e^{v^{j}_{n+1}-v^{j+1}_{n}}-e^{v^{j-1}_{n+1}-v^{j}_{n}}, \quad  \mbox{for}\quad |j|>1. 
\end{align*}
Now we take the limit for $\varepsilon\rightarrow0$ and obtain an equation $v^{0}_{n+1,x}=v^{0}_{n,x}$ and two semi-infinite lattices which are not related to each other:
\begin{align*}
&v^{1}_{n+1,x}=v^{1}_{n,x}+e^{v^{1}_{n+1}-v^{2}_{n}},\\
&v^{j}_{n+1,x}=v^{j}_{n,x}+e^{v^{j}_{n+1}-v^{j+1}_{n}}-e^{v^{j-1}_{n+1}-v^{j}_{n}}, \quad  \mbox{for}\quad j\geq2, 
\end{align*}
and
\begin{align*}
&v^{-1}_{n+1,x}=v^{-1}_{n,x}- e^{v^{-2}_{n+1}-v^{-1}_{n}},\\
&v^{j}_{n+1,x}=v^{j}_{n,x}+e^{v^{j}_{n+1}-v^{j+1}_{n}}-e^{v^{j-1}_{n+1}-v^{j}_{n}}, \quad  \mbox{for}\quad j\leq -2.
\end{align*}
By applying this manipulation to the lattice at two fixed points $j=-1$ and $j=N+1$ we obtain a finite field system of the form (\ref{eq00}) which is a desired reduction of the equation (\ref{eq1}):
\begin{align}
\begin{cases}\label{sys1}
u^{0}_{n+1,x}=u^{0}_{n,x}+e^{u^{0}_{n+1}-u^{1}_{n}},\\
u^{j}_{n+1,x}=u^{j}_{n,x}+e^{u^{j}_{n+1}-u^{j+1}_{n}}-e^{u^{j-1}_{n+1}-u^{j}_{n}}, \quad 1 < j < N-1, \\
u^{N}_{n+1,x}=u^{N}_{n,x}-e^{u^{N-1}_{n+1}-u^{N}_{n}}.
\end{cases}
\end{align}
In \cite{Smirnov2015} it was proved that system (\ref{sys1}) is integrable in the sense of Darboux for arbitrary natural $N$.
We rewrite the Lax pair (\ref{Laxeq1}) by means of the obtained boundary conditions 

\begin{align*} 
\begin{cases}
\varphi^{0}_{n+1}=-\varphi^{1}_{n}+e^{u^{0}_{n}-u^{0}_{n+1}}\varphi^{0}_{n}, \\
\varphi^{j}_{n+1}=-\varphi^{j+1}_{n}+e^{u^{j}_{n}-u^{j}_{n+1}}\varphi^{j}_{n}, \quad 1 < j < N-1, \\
\varphi^{N}_{n+1}=e^{u^{N}_{n}-u^{N}_{n+1}}\varphi^{N}_{n},
\end{cases} \quad
\begin{cases}
\varphi^{0}_{n,x}=0,\\
\varphi^{j}_{n,x}=-e^{u^{j-1}_{n}-u^{j}_{n}}\varphi^{j-1}_{n}, \quad 1 < j < N-1, \\
\varphi^{N}_{n,x}=-e^{u^{N-1}_{n}-u^{N}_{n}}\varphi^{N-1}_{n}.
\end{cases}
\end{align*}
When deriving it from (\ref{Laxeq1}) we set $\varphi^{-1}_n=0$,  $\varphi^{N+1}_n=0$. Let us study in more details the case $N=2$ that corresponds to the system
\begin{align}
\begin{cases}\label{sys1N2}
u^{0}_{n+1,x}=u^{0}_{n,x}+e^{u^{0}_{n+1}-u^{1}_{n}},\\
u^{1}_{n+1,x}=u^{1}_{n,x}+e^{u^{1}_{n+1}-u^{2}_{n}}-e^{u^{0}_{n+1}-u^{1}_{n}}, \\
u^{2}_{n+1,x}=u^{2}_{n,x}-e^{u^{1}_{n+1}-u^{2}_{n}}
\end{cases}
\end{align}
admitting the  Lax pair
\begin{align} \label{Laxsys1N2}
\begin{cases}
\varphi^{0}_{n+1}=-\varphi^{1}_{n}+e^{u^{0}_{n}-u^{0}_{n+1}}\varphi^{0}_{n}, \\
\varphi^{1}_{n+1}=-\varphi^{2}_{n}+e^{u^{1}_{n}-u^{1}_{n+1}}\varphi^{1}_{n}, \\
\varphi^{2}_{n+1}=e^{u^{2}_{n}-u^{2}_{n+1}}\varphi^{2}_{n},
\end{cases} \quad
\begin{cases}
\varphi^{0}_{n,x}=0,\\
\varphi^{1}_{n,x}=-e^{u^{0}_{n}-u^{1}_{n}}\varphi^{0}_{n}, \\
\varphi^{2}_{n,x}=-e^{u^{1}_{n}-u^{2}_{n}}\varphi^{1}_{n}.
\end{cases}
\end{align}
We give here $x$-integrals and $n$-integrals of system (\ref{sys1N2}) found earlier in  \cite{Smirnov2015}
\begin{align*} 
&J_1=e^{u^0_{n+2}-u^0_{n+3}}+e^{u^1_{n+1}-u^1_{n+2}}+e^{u^2_{n}-u^2_{n+1}};\\
&J_2=e^{u^0_{n+1}-u^0_{n+2}+u^1_{n+1}-u^1_{n+2}}+e^{u^0_{n+1}-u^0_{n+2}+u^2_{n}-u^2_{n+1}}+e^{u^1_{n+1}-u^1_{n+2}+u^2_{n}-u^2_{n+1}};\\
&J_3=e^{u^0_{n}-u^0_{n+1}+u^1_{n}-u^1_{n+1}+u^2_{n}-u^2_{n+1}}
\end{align*}
and, respectively
\begin{align*} 
&I_1=u^{0}_{n,x}+u^{1}_{n,x}+u^{2}_{n,x};\\
&I_2=u^{1}_{n,xx}+2u^{0}_{n,xx}+u^{0}_{n,x}u^{1}_{n,x}+u^{0}_{n,x}u^{2}_{n,x}+u^{1}_{n,x}u^{2}_{n,x};\\
&I_3=u^{0}_{n,xxx}+u^{0}_{n,x}u^{1}_{n,xx}+u^{0}_{n,xx}u^{1}_{n,x}+u^{0}_{n,xx}u^{2}_{n,x}+u^{0}_{n,x}u^{1}_{n,x}u^{2}_{n,x}.
\end{align*}

Note that they can be readily derived also from the Lax pair (\ref{Laxsys1N2}) due to Lemma~\ref{LemRk-n} and Lemma~\ref{LemP-x} above.

\subsection{Reductions of the equation (\ref{eq2})}
By applying the manipulations similar to that fulfilled in the previous section by using the same change of the variables one can show that lattice (\ref{eq2}) is reduced to the following finite field system which is of the form (\ref{eq00})
\begin{align*}
\begin{cases}
u^{0}_{n+1,x}=u^{0}_{n,x}-\frac{e^{u^{0}_{n}}}{e^{u^{1}_{n}}}+\frac{e^{u^{0}_{n+1}}}{e^{u^{1}_{n}}},\\
u^{j}_{n+1,x}=u^{j}_{n,x}+\frac{e^{u^{j-1}_{n+1}}}{e^{u^{j}_{n+1}}}-\frac{e^{u^{j-1}_{n+1}}}{e^{u^{j}_{n}}}-\frac{e^{u^{j}_{n}}}{e^{u^{j+1}_{n}}}+\frac{e^{u^{j}_{n+1}}}{e^{u^{j+1}_{n}}}, \quad 1 < j < N-1, \\
u^{N}_{n+1,x}=u^{N}_{n,x}+\frac{e^{u^{N-1}_{n+1}}}{e^{u^{N}_{n+1}}}-\frac{e^{u^{N-1}_{n+1}}}{e^{u^{N}_{n}}}
\end{cases}
\end{align*}
and its Lax pair is obtained from the system (\ref{Laxeq2}) by imposing additional conditions $\varphi^{-1}_n= 0$,  $\varphi^{N+1}_n=0$:
\begin{align*}
&\begin{cases}
\varphi^{0}_{n+1}=-e^{-u^{0}_{n+1}+u^{0}_{n}}\left(\varphi^{1}_{n}-\varphi^{0}_{n}\right)+\varphi^{1}_{n},\\
\varphi^{j}_{n+1}=-e^{-u^{j}_{n+1}+u^{j}_{n}}\left(\varphi^{j+1}_{n}-\varphi^{j}_{n}\right)+\varphi^{j+1}_{n},\quad 1 < j < N-1,\\
\varphi^{N}_{n+1}=e^{-u^{N}_{n+1}+u^{N}_{n}}\varphi^{N}_{n},
\end{cases}\\
&\begin{cases}
\varphi^{0}_{n,x}=0,\\
\varphi^{j}_{n,x}=e^{u^{j-1}_{n}-u^{j}_{n}}\left(\varphi^{j-1}_{n}-\varphi^{j}_{n}\right),\quad 1 < j < N-1,\\
\varphi^{N}_{n,x}=e^{u^{N-1}_{n}-u^{N}_{n}}\left(\varphi^{N-1}_{n}-\varphi^{N}_{n}\right).
\end{cases}
\end{align*}
We will investigate in more detail the case when $N=1$. Then we obtain a system
\begin{align}\label{sys2N1}
\begin{cases}
u^{0}_{n+1,x}=u^{0}_{n,x}-e^{u^{0}_{n}-u^{1}_{n}}+e^{u^{0}_{n+1}-u^{1}_{n}},\\
u^{1}_{n+1,x}=u^{1}_{n,x}+e^{u^{0}_{n+1}-u^{1}_{n+1}}-e^{u^{0}_{n+1}-u^{1}_{n}}
\end{cases}
\end{align}
and its Lax pair
\begin{align}\label{Laxsys2N1}
\begin{cases}
\varphi^{0}_{n+1}=-e^{-u^{0}_{n+1}+u^{0}_{n}}\left(\varphi^{1}_{n}-\varphi^{0}_{n}\right)+\varphi^{1}_{n},\\
\varphi^{1}_{n+1}=e^{-u^{1}_{n+1}+u^{1}_{n}}\varphi^{1}_{n},
\end{cases} \quad
\begin{cases}
\varphi^{0}_{n,x}=0,\\
\varphi^{1}_{n,x}=e^{u^{0}_{n}-u^{1}_{n}}\left(\varphi^{0}_{n}-\varphi^{1}_{n}\right).
\end{cases}
\end{align}
$x$-integrals and $n$-integrals of system (\ref{sys2N1}) have the form 
\begin{equation} \label{x-sys2N1}
\begin{aligned} 
&J_1=\left(e^{u^{1}_{n+1}-u^{1}_{n}}-1\right)\left(e^{u^{0}_{n+1}-u^{0}_{n}}-1\right);\\
&J_2=\left(e^{u^{1}_{n}-u^{1}_{n+1}}-1\right)\left(e^{u^{0}_{n+1}-u^{0}_{n+2}}-1\right)
\end{aligned} 
\end{equation}
and, respectively
\begin{equation} \label{n-sys2N1}
\begin{aligned} 
&I_1=u^{0}_{n,x}+u^{1}_{n,x}-e^{u^{0}_{n}-u^{1}_{n}};\\
&I_2=u^{0}_{n,xx}+u^{0}_{n,x}u^{1}_{n,x}-e^{u^{0}_{n}-u^{1}_{n}}u^{0}_{n,x}.
\end{aligned} 
\end{equation}

Let us briefly discuss the methods of searching the integrals. Note that for some cases the Lax pair gives a convenient tool for solving the problem, however this way is applied not always. Therefore in the other cases we use algebraic method. We explain both approaches with the example of the system (\ref{sys2N1}). For constructing the $x$-integrals we use the method of characteristic algebras and we use the Lax pair for finding the $n$-integrals.

Assume that function $H(u^{0}_{n},u^{1}_{n},u^{0}_{n+1},u^{1}_{n+1},\ldots)$ is an $x$--integral for the system (\ref{sys2N1}). Then according to the definition the relation 
\begin{equation*}\label{sys2DxH}
D_xH(u^{0}_{n},u^{1}_{n},u^{0}_{n+1},u^{1}_{n+1},\ldots)=0
\end{equation*}
should hold. Due to the chain rule equation (\ref{sys2DxH}) implies
\begin{equation*}\label{sys2K0H}
K_0H=0,
\end{equation*}
where 
\begin{equation*}\label{sys2K0}
K_0=u^{0}_{n,x}\frac{\partial}{\partial u_{n}^{0}}+u^{1}_{n,x}\frac{\partial}{\partial u_{n}^{1}}+u_{n+1,x}^{0}\frac{\partial}{\partial u_{n+1}^{0}}+u_{n+1,x}^{1}\frac{\partial}{\partial u_{n+1}^{1}}+ u_{n-1,x}^{0} \frac{\partial}{\partial u^{0}_{n-1}}+u_{n-1,x}^{1} \frac{\partial}{\partial u^{1}_{n-1}} +\ldots.
\end{equation*}
We exclude the terms like $u^{0}_{n\pm i,x}$, $u^{1}_{n\pm i,x}$ due to the system (\ref{sys2N1}) and get 
\begin{align*}\label{sys2K0fj}
K_0=&u^{0}_{n,x}\frac{\partial}{\partial u_{n}^{0}}+u^{1}_{n,x}\frac{\partial}{\partial u_{n}^{1}}+\left(u^{0}_{n,x}-e^{u^{0}_{n}-u^{1}_{n}}+e^{u^{0}_{n+1}-u^{1}_{n}}\right)\frac{\partial}{\partial u_{n+1}^{0}}\nonumber\\
&+\left(u^{1}_{n,x}+e^{u^{0}_{n+1}-u^{1}_{n+1}}-e^{u^{0}_{n+1}-u^{1}_{n}}\right)\frac{\partial}{\partial u_{n+1}^{1}} +\ldots.
\end{align*}
Obviously equation $K_0H=0$ is overdetermined since the coefficients of the equation depend on $u^{0}_{n,x}$ and $u^{1}_{n,x}$ while the solution $H$ does not depend on these variables, in other words we are interested on the solution $H$ of the equation (\ref{sys2K0H}) which solves in addition two more equations
\begin{equation*}\label{sys2X1X2}
X_1H=0 \quad \mbox{and} \quad X_2H=0,
\end{equation*}
where $X_1=\frac{\partial}{\partial u^{0}_{n,x}}$, $X_2=\frac{\partial}{\partial u^{1}_{n,x}}$. More precisely we get a system of the first order linear partial differential equations for one and the same unknown $H$
\begin{equation}\label{sys2K0X1X2}
\begin{aligned}
&K_0H=0, \\
&X_1H=0,  \\
&X_2H=0.  
\end{aligned}
\end{equation}
Evidently operator $K_0$ is represented as a linear combination of the vector fields
\begin{equation*}\label{sys2KY1Y2W}
K_0=u^{0}_{n,x}Y_1+u^{1}_{n,x}Y_2+W,
\end{equation*}
where the coefficients of the operators $Y_1$, $Y_2$, $W$ do not depend on the variables $u^{0}_{n,x}$, $u^{1}_{n,x}$:
\begin{align*}
&Y_1=\frac{\partial}{\partial u^{0}_{n}}+\frac{\partial}{\partial u^{0}_{n+1}}+\frac{\partial}{\partial u^{0}_{n+2}}, \\
&Y_2=\frac{\partial}{\partial u^{1}_{n}}+\frac{\partial}{\partial u^{1}_{n+1}},  \\
&W=\left(e^{u^{0}_{n+1}-u^{1}_{n}}-e^{u^{0}_{n}-u^{1}_{n}}\right)\frac{\partial}{\partial u^{0}_{n+1}}+
\left(e^{u^{0}_{n+1}-u^{1}_{n+1}}-e^{u^{0}_{n+1}-u^{1}_{n}}\right)\frac{\partial}{\partial u^{1}_{n+1}}\\
& \qquad + \left(e^{u^{0}_{n+2}-u^{1}_{n+1}}-e^{u^{0}_{n+1}-u^{1}_{n+1}}+e^{u^{0}_{n+1}-u^{1}_{n}}-e^{u^{0}_{n}-u^{1}_{n}}\right)\frac{\partial}{\partial u^{0}_{n+2}}. 
\end{align*}
Thus (\ref{sys2K0X1X2}) is reduced to the form
\begin{equation}\label{sys2Y1Y2W}
\begin{aligned}
&Y_1H=0, \\
&Y_2H=0,  \\
&WH=0.  
\end{aligned}
\end{equation}
It is checked that system (\ref{sys2Y1Y2W}) is closed, i.e. all of the commutators $\left[Y_1,Y_2\right]$, $\left[Y_1,W\right]$ and $\left[Y_2,W\right]$ are linearly expressed in terms of $Y_1$, $Y_2$, $W$ such that:
\begin{equation*}\label{sys2commY1Y2W}
\begin{aligned}
&\left[Y_1,Y_2\right]=0, \\
&\left[Y_1,W\right]=W,  \\
&\left[Y_2,W\right]=-W.  
\end{aligned}
\end{equation*}
Since the system contains three equations then in order to get two functionally independent solutions we look for a solution depending on five variables
\begin{equation*}\label{sys2H}
H=H\left(u^{0}_{n},u^{1}_{n},u^{0}_{n+1},u^{1}_{n+1},u^{0}_{n+2}\right).
\end{equation*}
Then system (\ref{sys2Y1Y2W}) turns into
\begin{align}\label{sys2sys1}
\begin{aligned}
&H_{u^{0}_{n}}+H_{u^{0}_{n+1}}+H_{u^{0}_{n+2}}=0, \\
&H_{u^{1}_{n}}+H_{u^{1}_{n+1}}=0,  \\
&\left(e^{u^{0}_{n+1}-u^{1}_{n}}-e^{u^{0}_{n}-u^{1}_{n}}\right)H_{u^{0}_{n+1}}+
\left(e^{u^{0}_{n+1}-u^{1}_{n+1}}-e^{u^{0}_{n+1}-u^{1}_{n}}\right)H_{u^{1}_{n+1}}\\
& \qquad + \left(e^{u^{0}_{n+2}-u^{1}_{n+1}}-e^{u^{0}_{n+1}-u^{1}_{n+1}}+e^{u^{0}_{n+1}-u^{1}_{n}}-e^{u^{0}_{n}-u^{1}_{n}}\right)H_{u^{0}_{n+2}}=0. 
\end{aligned}
\end{align}
In order to solve the system we have to reduce it due to Jacobi method to a normal (triangular) form
\begin{align*}
&\left(e^{u^{0}_{n+1}-u^{1}_{n}}-e^{u^{0}_{n}-u^{1}_{n}}\right)H_{u^{0}_{n}}-
\left(e^{u^{0}_{n+1}-u^{1}_{n+1}}-e^{u^{0}_{n+1}-u^{1}_{n}}\right)H_{u^{1}_{n+1}}\\
&\qquad \qquad \qquad -\left(e^{u^{0}_{n+2}-u^{1}_{n+1}}-e^{u^{0}_{n+1}-u^{1}_{n+1}}\right)H_{u^{0}_{n+2}}=0, \\
&H_{u^{1}_{n}}+H_{u^{1}_{n+1}}=0,  \\
&\left(e^{u^{0}_{n+1}-u^{1}_{n}}-e^{u^{0}_{n}-u^{1}_{n}}\right)H_{u^{0}_{n+1}}+
\left(e^{u^{0}_{n+1}-u^{1}_{n+1}}-e^{u^{0}_{n+1}-u^{1}_{n}}\right)H_{u^{1}_{n+1}}\\
& \qquad \qquad \qquad + \left(e^{u^{0}_{n+2}-u^{1}_{n+1}}-e^{u^{0}_{n+1}-u^{1}_{n+1}}+e^{u^{0}_{n+1}-u^{1}_{n}}-e^{u^{0}_{n}-u^{1}_{n}}\right)H_{u^{0}_{n+2}}=0. 
\end{align*}
Then by applying the consecutive integration according to the Jacobi algorithm we get the result (\ref{x-sys2N1}).

To construct the complete set of the $n$-integrals for the system (\ref{sys2N1}) we use the Lax pair. At first we make a change of the variables 
\begin{equation*}
\varphi^0_n=e^{-u^0_n}\psi^0_n, \quad \varphi^1_n=e^{-u^1_n}\psi^1_n
\end{equation*}
in the system (\ref{Laxsys2N1}) and get a new form of the Lax pair
\begin{equation}\label{Laxsys2UV}
\Psi_{n,x}=V_n\Psi_{n}, \quad \Psi_{n+1}=U_n\Psi_{n},
\end{equation}
where $\Psi_{n}=\left(\psi^0_n,\psi^1_n\right)^T$ and 
\begin{equation*}\label{sys2UV}
V_n=\left(\begin{array}{cc}
u^0_{n,x} & 0 \\
1 & u^1_{n,x}-e^{u^0_{n}-u^1_{n}}
\end{array} \right), \qquad
U_n=\left(\begin{array}{cc}
1 & e^{u^0_{n+1}-u^1_n}-e^{u^0_n-u^1_n} \\
0 & 1
\end{array} \right).
\end{equation*}
Now the Lax pair (\ref{Laxsys2UV}) satisfies all the requests of the Lemma~\ref{LemRk-n}. Therefore function $r^{(k)}_{N1,n}$ is an $n$-integral. In this case we have 
\begin{align*}
&r^{(2)}_{21,n}=u^{0}_{n,x}+u^{1}_{n,x}-e^{u^{0}_{n}-u^{1}_{n}},\\
&r^{(3)}_{21,n}=u^{0}_{n,xx}+2u^{1}_{n,xx}+(u^{0}_{n,x})^2+u^{0}_{n,x}u^{1}_{n,x}+(u^{1}_{n,x})^2-3e^{u^{0}_{n}-u^{1}_{n}}u^{0}_{n,x}+e^{2\left(u^{0}_{n}-u^{1}_{n}\right)}.
\end{align*}
After some slight simplifications we find integrals (\ref{n-sys2N1}).

\subsection{Reductions of the equation (\ref{eq3})}
In equation (\ref{eq3}) we put $u^0_n=c_0, \quad u_n^{N+1}=c_N:$
\begin{align*}
\begin{cases}
u^{1}_{n+1,x}=u^{1}_{n,x}\frac{\left(u^{1}_{n+1}\right)^2}{c_0u^{2}_{n}},\\
u^{j}_{n+1,x}=u^{j}_{n,x}\frac{\left(u^{j}_{n+1}\right)^2}{u^{j-1}_{n+1}u^{j+1}_{n}}, \quad 1 < j < N-1, \\
u^{N}_{n+1,x}=u^{N}_{n,x}\frac{\left(u^{N}_{n+1}\right)^2}{u^{N-1}_{n+1}c_N}.
\end{cases}
\end{align*}
We will rewrite the Lax pair (\ref{Laxeq3}) with the above constraint for the field variables and additional constraint for the eigenfunctions $\varphi^{-1}_n= 0$, $\varphi^{N+1}_n= 0$
\begin{align*}
\begin{cases}
\varphi^{0}_{n+1}=-\frac{c_0}{u^{1}_{n}}\left(\varphi^{0}_{n}-\varphi^{1}_{n}\right),\\
\varphi^{j}_{n+1}=-\frac{u^{j}_{n+1}}{u^{j+1}_{n}}\left(\varphi^{j}_{n}-\varphi^{j+1}_{n}\right),\quad 1 < j < N-1,\\
\varphi^{N}_{n+1}=-\frac{u^{N}_{n+1}}{c_N}\varphi^{N}_{n},
\end{cases}\quad 
\begin{cases}
\varphi^{0}_{n,x}=0,\\
\varphi^{j}_{n,x}=-\frac{u^{j}_{n,x}}{u^{j}_{n}}\left(\varphi^{j-1}_{n}-\varphi^{j}_{n}\right),\quad 1 < j < N-1,\\
\varphi^{N}_{n,x}=-\frac{u^{N}_{n,x}}{u^{N}_{n}}\left(\varphi^{N-1}_{n}-\varphi^{N}_{n}\right).
\end{cases}
\end{align*}
We put $N=1$. Then we obtain an equation
\begin{align}\label{sys3N1}
u^{1}_{n+1,x}=u^{1}_{n,x}\left(u^{1}_{n+1}\right)^2
\end{align}
having the Lax pair
\begin{align*}
\begin{cases}
\varphi^{0}_{n+1}=-\frac{1}{u^{1}_{n}}\left(\varphi^{0}_{n}-\varphi^{1}_{n}\right),\\
\varphi^{1}_{n+1}=-u^{1}_{n+1}\varphi^{1}_{n},
\end{cases} \quad
\begin{cases}
\varphi^{0}_{n,x}=0,\\
\varphi^{1}_{n,x}=-\frac{u^{1}_{n,x}}{u^{1}_{n}}\left(\varphi^{0}_{n}-\varphi^{1}_{n}\right).
\end{cases}
\end{align*}
$x$-integrals and $n$-integrals of equation (\ref{sys3N1}) have the form 
\begin{align*} 
J=u^1_n+\frac{1}{u^1_{n+1}}
\end{align*}
and, respectively
\begin{align*} 
I=\frac{u^{1}_{n,xxx}}{u^{1}_{n,x}}-\frac{3}{2}\frac{(u^{1}_{n,xx})^2}{(u^{1}_{n,x})^2}.
\end{align*}

\subsection{Reductions of the equation (\ref{eq4})}
In the equation (\ref{eq4}) we put $u^0_n=c_1$, $u_n^{N+1}=c_N$ and obtain a system
\begin{align*}
\begin{cases}
u^{1}_{n+1,x}=u^{1}_{n,x}\frac{\left(u^{1}_{n+1}-u^{2}_{n}\right)}{c_1-u^{1}_{n}},\\
u^{j}_{n+1,x}=u^{j}_{n,x}\frac{\left(u^{j}_{n+1}-u^{j+1}_{n}\right)}{\left(u^{j-1}_{n+1}-u^{j}_{n}\right)}, \quad 2 < j < N-1, \\
u^{N}_{n+1,x}=u^{N}_{n,x}\frac{u^{N}_{n+1}-c_N}{\left(u^{N-1}_{n+1}-u^{N}_{n}\right)}
\end{cases}
\end{align*}
with the Lax pair
\begin{align*}
\begin{cases}
\varphi^{0}_{n+1}=\left(u^{1}_{n}-c_1\right)\varphi^{0}_{n}+\varphi^{1}_{n},\\
\varphi^{j}_{n+1}=\left(u^{j+1}_{n}-u^{j}_{n+1}\right)\varphi^{j}_{n}+\varphi^{j+1}_{n},\quad 1 < j < N-1,\\
\varphi^{N}_{n+1}=\left(c_N-u^{N}_{n+1}\right)\varphi^{N}_{n},
\end{cases}\quad
\begin{cases}
\varphi^{0}_{n,x}=0,\\
\varphi^{j}_{n,x}=-u^{j}_{n,x}\varphi^{j-1}_{n},\quad 1 < j < N-1,\\
\varphi^{N}_{n,x}=-u^{N}_{n,x}\varphi^{N-1}_{n}.
\end{cases}
\end{align*}
In the particular case $N=2$ we get a system 
\begin{align}\label{sys4N2}
\begin{cases}
u^{1}_{n+1,x}=u^{1}_{n,x}\frac{\left(u^{1}_{n+1}-u^{2}_{n}\right)}{c_1-u^{1}_{n}},\\
u^{2}_{n+1,x}=u^{2}_{n,x}\frac{u^{2}_{n+1}-c_2}{\left(u^{1}_{n+1}-u^{2}_{n}\right)}
\end{cases}
\end{align}
with the Lax pair
\begin{align}\label{Laxsys4N2}
\begin{cases}
\varphi^{0}_{n+1}=\left(u^{1}_{n}-c_1\right)\varphi^{0}_{n}+\varphi^{1}_{n},\\
\varphi^{1}_{n+1}=\left(u^{2}_{n}-u^{1}_{n+1}\right)\varphi^{1}_{n}+\varphi^{2}_{n},\\
\varphi^{2}_{n+1}=\left(c_2-u^{2}_{n+1}\right)\varphi^{2}_{n},
\end{cases} \quad
\begin{cases}
\varphi^{0}_{n,x}=0,\\
\varphi^{1}_{n,x}=-u^{1}_{n,x}\varphi^{0}_{n},\\
\varphi^{2}_{n,x}=-u^{2}_{n,x}\varphi^{1}_{n}.
\end{cases}
\end{align}
$x$-integrals and $n$-integrals of system (\ref{sys4N2}) have the form 
\begin{equation} \label{sys4-x-int}
\begin{aligned} 
&J_1=\left(u^{1}_{n}-c_1\right)\left(u^{2}_{n+1}-c_2\right)\left(u^{2}_{n}-u^{1}_{n+1}\right);\\
&J_2=u^{1}_{n+1}u^{1}_{n+2}-c_1u^{1}_{n+2}-u^{1}_{n+1}u^{2}_{n+1}-c_2u^{2}_{n}+u^{2}_{n}u^{2}_{n+1}
\end{aligned}
\end{equation}
and, respectively
\begin{equation} \label{sys4-n-int}
\begin{aligned} 
I_1=&\frac{u^{1}_{n,xxx}}{u^{1}_{n,x}}+\frac{u^{2}_{n,xxx}}{u^{2}_{n,x}}-\frac{4}{3}\left(\frac{u^{1}_{n,xx}}{u^{1}_{n,x}}\right)^2-\frac{1}{3}\frac{u^{1}_{n,xx}u^{2}_{n,xx}}{u^{1}_{n,x}u^{2}_{n,x}}-\frac{4}{3}\left(\frac{u^{2}_{n,xx}}{u^{2}_{n,x}}\right)^2;\\
I_2=&\frac{u^{1}_{n,xxxx}}{u^{1}_{n,x}}
-\frac{13}{3}\frac{u^{1}_{n,xx}u^{1}_{n,xxx}}{\left(u^{1}_{n,x}\right)^2}-\frac{2}{3}\frac{u^{2}_{n,xx}u^{1}_{n,xxx}}{u^{1}_{n,x}u^{2}_{n,x}}
+\frac{2}{3}\frac{u^{2}_{n,xx}u^{2}_{n,xxx}}{\left(u^{2}_{n,x}\right)^2}
+\frac{1}{3}\frac{u^{1}_{n,xx}u^{2}_{n,xxx}}{u^{1}_{n,x}u^{2}_{n,x}}\\
&+\frac{32}{9}\left(\frac{u^{1}_{n,xx}}{u^{1}_{n,x}}\right)^3
+\frac{\left(u^{1}_{n,xx}\right)^2u^{2}_{n,xx}}{\left(u^{1}_{n,x}\right)^2u^{2}_{n,x}}
-\frac{2}{3}\frac{u^{1}_{n,xx}\left(u^{2}_{n,xx}\right)^2}{u^{1}_{n,x}\left(u^{2}_{n,x}\right)^2}
-\frac{8}{9}\left(\frac{u^{2}_{n,xx}}{u^{2}_{n,x}}\right)^3.
\end{aligned}
\end{equation}

Now we briefly explain how the integrals (\ref{sys4-x-int}) and (\ref{sys4-n-int}) were constructed. Let us begin with $x$-integrals. The Lax pair can be written as 
\begin{equation}\label{Laxsys4UV}
\Psi_{n,x}=V_n\Psi_{n}, \quad \Psi_{n+1}=U_n\Psi_{n},
\end{equation}
where  
\begin{equation*}\label{sys4UV}
V_n=\left(\begin{array}{ccc}
0 & 0 & 0\\
-u^1_{n,x} & 0 & 0\\
0 & -u^2_{n,x} & 0
\end{array} \right), \qquad
U_n=\left(\begin{array}{ccc}
u^1_n-c_1 & 1 & 0 \\
0 & u^2_n-u^1_{n+1} & 1\\
0 & 0 & c_2-u^2_{n+1}
\end{array} \right).
\end{equation*}
Since the Lax pair (\ref{Laxsys4UV}) satisfies the conditions of Lemma~\ref{LemP-x}, we can use it for deriving the $x$-integrals of the system (\ref{sys4N2}). To this end we evaluate the product
\begin{equation*}
P^{(k)}_n=U_{n+k} U_{n+k-1}\cdots U_{n}
\end{equation*}
and take its entry located at the right upper corner. It is easily checked that for $k=1$ and $k=2$ we get trivial integrals $J=1$ and $J=c_2-c_1$. For $k=3$ and $k=4$ we obtain 
\begin{align*}
p^{(3)}_{13,n}=&\left(u^1_{n+3}-c_1\right)\left(u^1_{n+2}-c_1\right)+\left(u^2_{n+2}-c_1\right)\left(u^1_{n+2}-c_2\right)+\left(u^2_{n+2}-c_2\right)\left(u^2_{n+1}-c_2\right),\\
p^{(4)}_{13,n}=&\left(u^1_{n+3}-c_1\right)\left(u^1_{n+2}-c_1\right)\left(u^1_{n+1}-c_1\right)-\left(u^1_{n+3}-c_1\right)\left(u^2_{n+1}-c_1\right)\left(u^1_{n+1}-c_2\right) \\
&-\left(u^1_{n+3}-u^2_{n+2}\right)\left(u^1_{n+2}-u^2_{n+1}\right)\left(u^1_{n+1}-c_2\right)+\left(c_2-c_1\right)\left(u^2_{n+1}-c_2\right)\left(u^2_{n}-c_2\right).
\end{align*}
After some slight simplifications we reduce them into integrals $J_1$ and $J_2$ given in (\ref{sys4-x-int}).

The next step is to explain the algebraic method of looking for the $n$-integrals. Let us emphasize that this task is more difficult since the algebra $L_n$ has a complicated structure. 

It is convenient to introduce new notations by setting $u_n:=u^1_n$ and $v_n:=u^2_n$. Thus we convert the system (\ref{sys4N2}) to the form
\begin{align}\label{sys4N2uv}
\begin{cases}
u_{n+1,x}=F^1,\\
v_{n+1,x}=F^2,
\end{cases}
\end{align}
where
\begin{equation*}
F^1=u_{n,x}\frac{\left(u_{n+1}-v_{n}\right)}{c_1-u_{n}}, \qquad
F^2=v_{n,x}\frac{v_{n+1}-c_2}{\left(u_{n+1}-v_{n}\right)}.
\end{equation*}
At first we evaluate the subalgebra $L^{(1)}_n$, generated by the vector fields $Y_{1,1}$, $Y_{2,1}$, $X_{1,1}$ and $X_{2,1}$, where
\begin{align*}
Y_{1,1}=&\frac{\partial}{\partial u_{n}}+D_n^{-1}\left(\frac{\partial F_n^1}{\partial u_{n+1}}\right)\frac{\partial}{\partial u_{n,x}}+D_n^{-1}\left(\frac{\partial F_n^2}{\partial u_{n+1}}\right)\frac{\partial}{\partial v_{n,x}}\\
&+D_n^{-1}\left(\frac{\partial F_{n,x}^1}{\partial u_{n+1}}\right)\frac{\partial}{\partial u_{n,xx}}+D_n^{-1}\left(\frac{\partial F_{n,x}^2}{\partial u_{n+1}}\right)\frac{\partial}{\partial v_{n,xx}}+\dots,\\
Y_{2,1}=&\frac{\partial}{\partial v_{n}}+D_n^{-1}\left(\frac{\partial F_n^1}{\partial v_{n+1}}\right)\frac{\partial}{\partial u_{n,x}}+D_n^{-1}\left(\frac{\partial F_n^2}{\partial v_{n+1}}\right)\frac{\partial}{\partial v_{n,x}}\\
&+D_n^{-1}\left(\frac{\partial F_{n,x}^1}{\partial v_{n+1}}\right)\frac{\partial}{\partial u_{n,xx}}+D_n^{-1}\left(\frac{\partial F_{n,x}^2}{\partial v_{n+1}}\right)\frac{\partial}{\partial v_{n,xx}}+\dots,\\
X_{1,1}=&\frac{\partial}{\partial u_{n-1}}, \qquad X_{2,1}=\frac{\partial}{\partial v_{n-1}}.
\end{align*}
Explicit form of the operators $Y_{1,1}$ and $Y_{2,1}$ are
\begin{align*}
Y_{1,1}=&\frac{\partial}{\partial u_{n}}+\frac{u_{n,x}}{u_n-v_{n-1}}\frac{\partial}{\partial u_{n,x}}-\frac{v_{n,x}}{u_n-v_{n-1}}\frac{\partial}{\partial v_{n,x}}\nonumber\\
&+\frac{1}{u_n-v_{n-1}}\left(u_{n,xx}+\frac{u_{n,x}v_{n,x}}{v_n-c_2}\right)\frac{\partial}{\partial u_{n,xx}}-\frac{1}{u_n-v_{n-1}}\left(v_{n,xx}+\frac{2v^2_{n,x}}{v_n-c_2}\right)\frac{\partial}{\partial v_{n,xx}}+\dots,\\
Y_{2,1}=&\frac{\partial}{\partial v_{n}}+\frac{v_{n,x}}{v_n-c_2}\frac{\partial}{\partial v_{n,x}}+\frac{v_{n,xx}}{v_n-c_2}\frac{\partial}{\partial v_{n,xx}}+\frac{v_{n,xxx}}{v_n-c_2}\frac{\partial}{\partial v_{n,xxx}}\dots.
\end{align*}
It can be verified by analyzing the further terms in the series for $Y_{1,1}$ that it is decomposed into a sum
\begin{equation*}
Y_{1,1}=Y_1+\frac{1}{u_n-v_{n-1}}Z_1,
\end{equation*}
of two operators belonging the algebra $L^{(1)}_n$, since $Z_1=(u_n-v_{n-1})^2[X_{2,1}, Y_{1,1}]$. Here $Y_1=\frac{\partial}{\partial u_{n}}$ and the coefficients in 
\begin{equation*}
Z_1=u_{n,x}\frac{\partial}{\partial u_{n,x}}-v_{n,x}\frac{\partial}{\partial v_{n,x}}
+\left(u_{n,xx}+\frac{u_{n,x}v_{n,x}}{v_n-c_2}\right)\frac{\partial}{\partial u_{n,xx}}-\left(v_{n,xx}+\frac{2v^2_{n,x}}{v_n-c_2}\right)\frac{\partial}{\partial v_{n,xx}}+\dots
\end{equation*}
do not depend on the variables $u_{n-1}$, $v_{n-1}$. Therefore operators $Y_1$, $Z_1$ and $Y_{2,1}$ commute with $X_{1,1}$ and $X_{2,1}$. Moreover the operators $X_{1,1}$, $X_{2,1}$, $Y_1$, $Z_1$ and $Y_{2,1}$ constitute a basis in the algebra $L_n^{(1)}$. Thus algebra $L_n^{(1)}$ is of dimension 5. However solution of the system 
\begin{equation*}
X_{1,1}I=0,\quad X_{2,1}I=0, \quad Y_1I=0, \quad Z_1I=0, \quad Y_{2,1}I=0
\end{equation*}
does not produce any $n$--integral for the system (\ref{sys4N2uv}), since the relation $D_nI=I$ is not satisfied.

Therefore, we need to pass to the consideration of the second subalgebra $L_n^{(2)}$ of the characteristic algebra $L_n$ generated by the operators
\begin{align*}
X_{1,1},\quad X_{2,1}, \quad X_{1,2},\quad X_{2,2}, \quad Y_{1,1}, \quad Y_{2,1}, \quad Y_{1,2}, \quad Y_{2,2}. 
\end{align*}
In other words we add to $L_n^{(1)}$ four extra operators 
\begin{equation*}
X_{1,2}=\frac{\partial}{\partial u_{n-2}},\quad X_{2,2}=\frac{\partial}{\partial v_{n-2}}, \quad Y_{1,2}=D_n^{-1}Y_{1,1}D_n, \quad Y_{2,2}=D_n^{-1}Y_{2,1}D_n. 
\end{equation*}
Two last operators are of the form
\begin{align*}
Y_{1,2}=&\frac{u_{n,x}(v_{n-2}-c_1)}{(u_{n-1}-v_{n-2})(u_{n-1}-c_1)}\frac{\partial}{\partial u_{n,x}}-\frac{v_{n,x}}{u_{n-1}-v_{n-2}}\frac{\partial}{\partial v_{n,x}}\nonumber\\
&+\frac{1}{u_{n-1}-v_{n-2}}\left(\frac{u_{n,xx}(v_{n-2}-c_1)}{u_{n-1}-c_1}+\frac{2u^2_{n,x}(v_{n-2}-c_1)}{(u_{n-1}-c_1)(u_n-v_{n-1})}+\frac{u_{n,x}v_{n,x}(u_n-c_2)}{(v_{n-1}-c_2)(v_n-c_2)}\right)\frac{\partial}{\partial u_{n,xx}}\\
&-\frac{1}{u_{n-1}-v_{n-2}}\left(v_{n,xx}+\frac{2v^2_{n,x}(u_n-c_2)}{(v_n-c_2)(v_{n-1}-c_2)}+\frac{u_{n,x}v_{n,x}(v_{n-2}-c_1)}{(u_{n-1}-c_1)(u_n-v_{n-1})}\right)\frac{\partial}{\partial v_{n,xx}}+\dots,\\
Y_{2,2}=&-\frac{u_{n,x}}{u_n-v_{n-1}}\frac{\partial}{\partial u_{n,x}}+\frac{v_{n,x}(u_n-c_2)}{(u_{n}-v_{n-1})(v_{n-1}-c_2)}\frac{\partial}{\partial v_{n,x}}\\
&-\frac{1}{u_n-v_{n-1}}\left(u_{n,xx}+\frac{u_{n,x}v_{n,x}(u_n-c_2)}{(v_n-c_2)(v_{n-1}-c_2)}\right)\frac{\partial}{\partial u_{n,xx}}\\
&+\frac{u_n-c_2}{(u_n-v_{n-1})(v_{n-1}-c_2)}\left(v_{n,xx}+\frac{2v^2_{n,x}}{v_n-c_2}\right)\frac{\partial}{\partial v_{n,xx}}+\dots.
\end{align*}
It can be proved that these operators are represented as linear combinations of the operators
\begin{align*}
&Z_2=v_{n,x}\frac{\partial}{\partial v_{n,x}}-
\frac{u_{n,x}v_{n,x}}{v_n-c_2}\frac{\partial}{\partial u_{n,xx}}
+\left(v_{n,xx}+\frac{2v^2_{n,x}}{v_n-c_2}\right)\frac{\partial}{\partial v_{n,xx}}+\dots,\\
&Y_3=u_{n,x}\frac{\partial}{\partial u_{n,x}}+u_{n,xx}\frac{\partial}{\partial u_{n,xx}}+u_{n,xxx}\frac{\partial}{\partial u_{n,xxx}}+\dots,\\
&Z_3=2u^2_{n,x}\frac{\partial}{\partial u_{n,xx}}-u_{n,x}v_{n,x}\frac{\partial}{\partial v_{n,xx}}+\left(6u_{n,x}u_{n,xx}+\frac{3u^2_{n,x}v_{n,x}}{v_n-c_2}\right)\frac{\partial}{\partial u_{n,xxx}}+\dots,\\
&Z_4=\frac{3}{2}v_{n,x}\frac{\partial}{\partial v_{n,xx}}+\left(u_{n,xx}-\frac{u_{n,x}v_{n,xx}}{v_{n,x}}+\frac{3}{2}\frac{u_{n,x}v_{n,x}}{v_n-c_2}\right)\frac{\partial}{\partial u_{n,xxx}}+\dots,\\
&Y_4=3u_{n,x}v_{n,x}\frac{\partial}{\partial u_{n,xxx}}-3v^2_{n,x}\frac{\partial}{\partial v_{n,xxx}}+\left(14v_{n,x}u_{n,xx}+4u_{n,x}v_{n,xx}\right)\frac{\partial}{\partial u_{n,xxxx}}+\dots
\end{align*}
which are simpler, since they commute with $X_{1,1}$, $X_{1,2}$, $X_{2,1}$, $X_{2,2}$.

We prove that system of equations 
\begin{equation}\label{sys4sysI}
Y_1I=0, \quad Y_{2,1}I=0, \quad Y_3I=0, \quad Z_2I=0, \quad Z_3I=0, \quad Z_4I=0, \quad Y_4I=0
\end{equation}
is closed or in other words algebra $L^{(2)}_n$ is of finite dimension. We look for the solution of (\ref{sys4sysI}) in the form
\begin{equation*}
I=I\left(u_n,v_n,u_{n,x},v_{n,x},u_{n,xx},v_{n,xx},u_{n,xxx},v_{n,xxx},u_{n,xxxx}\right).
\end{equation*}
We reduce it to normal form and solve by Jacobi method. As a result we get (\ref{sys4-n-int}), where one has to replace $u_n:=u^1_n$ and $v_n:=u^2_n$.

\subsection{Reductions of the equation (\ref{eq5})}
In equation (\ref{eq5}) we put $u^0_n=c_0, \quad u_n^{N+1}=c_N:$
\begin{align*}
\begin{cases}
u^{1}_{n+1,x}=u^{1}_{n,x}\frac{u^{1}_{n+1}\left(u^{1}_{n+1}-u^{2}_{n}\right)}{u^{2}_{n}\left(c_{0}-u^{1}_{n}\right)},\\
u^{j}_{n+1,x}=u^{j}_{n,x}\frac{u^{j}_{n+1}\left(u^{j}_{n+1}-u^{j+1}_{n}\right)}{u^{j+1}_{n}\left(u^{j-1}_{n+1}-u^{j}_{n}\right)}, \quad 2 < j < N-1, \\
u^{N}_{n+1,x}=u^{N}_{n,x}\frac{u^{N}_{n+1}\left(u^{N}_{n+1}-c_N\right)}{c_N\left(u^{N-1}_{n+1}-u^{N}_{n}\right)}.
\end{cases}
\end{align*}
We rewrite the Lax pair (\ref{Laxeq5}) with the above restrictions and by setting the conditions $\varphi^{-1}_n= 0$, $\varphi^{N+1}_n=0$
\begin{align*}
&\begin{cases}
\varphi^{0}_{n+1}=\left(1-\frac{c_0}{u^{1}_{n}}\right)\varphi^{0}_{n}-\frac{c_0}{u^{1}_{n}}\varphi^{1}_{n},\\
\varphi^{j}_{n+1}=\left(1-\frac{u^{j}_{n+1}}{u^{j+1}_{n}}\right)\varphi^{j}_{n}-\frac{u^{j}_{n+1}}{u^{j+1}_{n}}\varphi^{j+1}_{n},\quad 1 < j < N-1,\\
\varphi^{N}_{n+1}=\left(1-\frac{u^{N}_{n+1}}{c_N}\right)\varphi^{N}_{n},
\end{cases} \\
&\begin{cases}
\varphi^{0}_{n,x}=0,\\
\varphi^{j}_{n,x}=\varphi^{j}_{n,x}=\frac{u^{j}_{n,x}}{u^{j}_{n}}\left(\varphi^{j-1}_{n}+\varphi^{j}_{n}\right),\quad 1 < j < N-1,\\
\varphi^{N}_{n,x}=\frac{u^{N}_{n,x}}{u^{N}_{n}}\left(\varphi^{N-1}_{n}+\varphi^{N}_{n}\right).
\end{cases}
\end{align*}
We concentrate on a simplest case by taking $N=1$. Then we get an equation
\begin{align}\label{sys5N1}
u^{1}_{n+1,x}=u^{1}_{n,x}\frac{u^{1}_{n+1}\left(u^{1}_{n+1}-c_1\right)}{c_1\left(c_0-u^{1}_{n}\right)},
\end{align}
and its Lax pair
\begin{align*}
\begin{cases}
\varphi^{0}_{n+1}=\left(1-\frac{c_0}{u^{1}_{n}}\right)\varphi^{0}_{n}-\frac{c_0}{u^{1}_{n}}\varphi^{1}_{n},\\
\varphi^{1}_{n+1}=\left(1-\frac{u^{1}_{n+1}}{c_1}\right)\varphi^{1}_{n},
\end{cases} \quad
\begin{cases}
\varphi^{0}_{n,x}=0,\\
\varphi^{1}_{n,x}=\frac{u^{1}_{n,x}}{u^{1}_{n}}\left(\varphi^{0}_{n}+\varphi^{1}_{n}\right).
\end{cases}
\end{align*}
$x$-integrals and $n$-integrals of equation (\ref{sys5N1}) have the form 
\begin{equation*} 
J=\frac{u^1_n}{u^1_{n+1}}+\frac{c_0}{u^1_{n+1}}+\frac{u^1_n}{c_1}
\end{equation*}
and, respectively
\begin{equation*} 
I=\frac{u^{1}_{n,xxx}}{u^{1}_{n,x}}-\frac{3}{2}\frac{(u^{1}_{n,xx})^2}{(u^{1}_{n,x})^2}.
\end{equation*}

\subsection{Reductions of the equation (\ref{eq6})}
In equation (\ref{eq6}) we put $u^0_n=c_0$, $u_n^{N+1}=c_N$:
\begin{align*}
\begin{cases}
u^{1}_{n+1,x}=u^{1}_{n,x}\frac{\left(c_0-u^{1}_{n+1}\right)\left(u^{1}_{n+1}-u^{2}_{n}\right)}{\left(c_0-u^{1}_{n}\right)\left(u^{1}_{n}-u^{2}_{n}\right)},\\
u^{j}_{n+1,x}=u^{j}_{n,x}\frac{\left(u^{j-1}_{n+1}-u^{j}_{n+1}\right)\left(u^{j}_{n+1}-u^{j+1}_{n}\right)}{\left(u^{j-1}_{n+1}-u^{j}_{n}\right)\left(u^{j}_{n}-u^{j+1}_{n}\right)}, \quad 2 < j < N-1, \\
u^{N}_{n+1,x}=u^{N}_{n,x}\frac{\left(u^{N-1}_{n+1}-u^{N}_{n+1}\right)\left(u^{N}_{n+1}-c_N\right)}{\left(u^{N-1}_{n+1}-u^{N}_{n}\right)\left(u^{N}_{n}-c_N\right)}.
\end{cases}
\end{align*}
We rewrite the Lax pair (\ref{Laxeq6}) with the constraint $u^0_n=c_0$, $u_n^{N+1}=c_N$, $\varphi^{-1}_n=0$, $\varphi^{N+1}_n=0$ 
\begin{align*}
&\begin{cases}
\varphi^{0}_{n+1}=\frac{c_0-u^{1}_{n}}{c_0-u^{1}_{n}}\varphi^{0}_{n}+\left(1-\frac{c_0-u^{1}_{n}}{c_0-u^{1}_{n}}\right)\varphi^{1}_{n},\\
\varphi^{j}_{n+1}=\frac{u^{j}_{n+1}-u^{j+1}_{n}}{u^{j}_{n}-u^{j+1}_{n}}\varphi^{j}_{n}+\left(1-\frac{u^{j}_{n+1}-u^{j+1}_{n}}{u^{j}_{n}-u^{j+1}_{n}}\right)\varphi^{j+1}_{n}, \quad 1 < j < N-1,\\
\varphi^{N}_{n+1}=\frac{u^{N}_{n+1}-c_N}{u^{N}_{n}-c_N}\varphi^{N}_{n},
\end{cases}\\
&\begin{cases}
\varphi^{0}_{n,x}=0,\\
\varphi^{j}_{n,x}=\frac{u^{j}_{n,x}}{u^{j-1}_{n}-u^{j}_{n}}\left(\varphi^{j-1}_{n}+\varphi^{j}_{n}\right),\quad 1 < j < N-1,\\
\varphi^{N}_{n,x}=\frac{u^{N}_{n,x}}{u^{N-1}_{n}-u^{N}_{n}}\left(\varphi^{N-1}_{n}+\varphi^{N}_{n}\right),
\end{cases}
\end{align*}
By putting $N=2$ we arrive at a system
\begin{align}\label{sys6N2}
\begin{cases}
u^{1}_{n+1,x}=u^{1}_{n,x}\frac{\left(c_0-u^{1}_{n+1}\right)\left(u^{1}_{n+1}-u^{2}_{n}\right)}{\left(c_0-u^{1}_{n}\right)\left(u^{1}_{n}-u^{2}_{n}\right)},\\
u^{2}_{n+1,x}=u^{2}_{n,x}\frac{\left(u^{1}_{n+1}-u^{2}_{n+1}\right)\left(u^{2}_{n+1}-c_N\right)}{\left(u^{1}_{n+1}-u^{2}_{n}\right)\left(u^{2}_{n}-c_2\right)}.
\end{cases}
\end{align}
admitting the Lax pair
\begin{align*}
\begin{cases}
\varphi^{0}_{n+1}=\frac{c_0-u^{1}_{n}}{c_0-u^{1}_{n}}\varphi^{0}_{n}+\left(1-\frac{c_0-u^{1}_{n}}{c_0-u^{1}_{n}}\right)\varphi^{1}_{n},\\
\varphi^{1}_{n+1}=\frac{u^{1}_{n+1}-u^{2}_{n}}{u^{1}_{n}-u^{2}_{n}}\varphi^{1}_{n}+\left(1-\frac{u^{1}_{n+1}-u^{2}_{n}}{u^{1}_{n}-u^{2}_{n}}\right)\varphi^{2}_{n},\\
\varphi^{2}_{n+1}=\frac{u^{2}_{n+1}-c_2}{u^{2}_{n}-c_2}\varphi^{2}_{n},
\end{cases} \quad 
\begin{cases}
\varphi^{0}_{n,x}=0,\\
\varphi^{1}_{n,x}=\frac{u^{1}_{n,x}}{u^{0}_{n}-u^{1}_{n}}\left(\varphi^{0}_{n}+\varphi^{1}_{n}\right),\\
\varphi^{2}_{n,x}=\frac{u^{2}_{n,x}}{u^{1}_{n}-u^{2}_{n}}\left(\varphi^{1}_{n}+\varphi^{2}_{n}\right).
\end{cases}
\end{align*}
$x$-integrals and $n$-integrals of system (\ref{sys6N2}) have the form 
\begin{align*} 
&J_1=\frac{\left(u^{1}_n-c_0\right)\left(u^{2}_n-c_2\right)\left(u^{1}_{n+1}-u^{2}_n\right)}{\left(u^{1}_{n+1}-u^{1}_n\right)\left(u^{2}_{n+1}-u^{2}_n\right)};\\
&J_2=\frac{\left(u^{1}_{n+2}-c_0\right)\left(u^{2}_n-c_2\right)\left(u^{2}_{n+1}-u^{1}_{n+1}\right)\left(u^{2}_{n+2}-u^{2}_{n+1}\right)}{\left(u^{1}_{n+1}-c_0\right)\left(u^{2}_{n+2}-c_2\right)\left(u^{2}_{n+1}-u^{2}_n\right)\left(u^{1}_{n+2}-u^{2}_{n+1}\right)}
\end{align*}
and, respectively
\begin{align*} 
&I_1=\frac{u^1_{n,x}u^2_{n,x}}{\left(u^{1}_n-c_0\right)\left(u^{2}_n-c_2\right)\left(u^{1}_n-u^{2}_n\right)};\\
&I_2=\frac{u^1_{n,xx}}{u^1_{n,x}}-\frac{2u^1_{n,x}}{u^{1}_n-c_0}+\frac{u^2_{n,x}\left(u^{1}_n-c_2\right)}{\left(u^{2}_n-c_2\right)\left(u^{1}_n-u^{2}_n\right)}.
\end{align*}

\subsection{Reductions of the equation (\ref{eq7})}
In the equations (\ref{eq7}), (\ref{Laxeq7}) we put $u^0_n=c_0$,  $u_n^{N+1}=c_N$,  $\varphi^{-1}_n=0$,  $\varphi^{N+1}_n=0$ and get a system 
\begin{align*}
\begin{cases}
u^{1}_{n+1,x}=u^{1}_{n,x}\frac{\sinh\left(c_0-u^{1}_{n+1}\right)\sinh\left(u^{1}_{n+1}-u^{2}_{n}\right)}{\sinh\left(c_0-u^{1}_{n}\right)\sinh\left(u^{1}_{n}-u^{2}_{n}\right)},\\
u^{j}_{n+1,x}=u^{j}_{n,x}\frac{\sinh\left(u^{j-1}_{n+1}-u^{j}_{n+1}\right)\sinh\left(u^{j}_{n+1}-u^{j+1}_{n}\right)}{\sinh\left(u^{j-1}_{n+1}-u^{j}_{n}\right)\sinh\left(u^{j}_{n}-u^{j+1}_{n}\right)},\quad 1 < j < N-1,\\
u^{N}_{n+1,x}=u^{N}_{n,x}\frac{\sinh\left(u^{N-1}_{n+1}-u^{N}_{n+1}\right)\sinh\left(u^{N}_{n+1}-c_N\right)}{\sinh\left(u^{N-1}_{n+1}-u^{N}_{n}\right)\sinh\left(u^{N}_{n}-c_N\right)}.
\end{cases}
\end{align*}
with the Lax pair
\begin{align*}
&\begin{cases}
\varphi^{0}_{n+1}=\frac{e^{2\left(c_0-u^{1}{n}\right)}-1}{e^{2\left(c_0-u^{1}_{n}\right)}-1}\varphi^{0}_{n}+\left(1-\frac{e^{2\left(c_0-u^{1}_{n}\right)}-1}{e^{2\left(c_0-u^{1}_{n}\right)}-1}\right)\varphi^{1}_{n},\\
\varphi^{j}_{n+1}=\frac{e^{2\left(u^{j}_{n+1}-u^{j+1}{n}\right)}-1}{e^{2\left(u^{j}_{n}-u^{j+1}_{n}\right)}-1}\varphi^{j}_{n}+\left(1-\frac{e^{2\left(u^{j}_{n+1}-u^{j+1}_{n}\right)}-1}{e^{2\left(u^{j}_{n}-u^{j+1}_{n}\right)}-1}\right)\varphi^{j+1}_{n}, \quad 1 < j < N-1,\\
\varphi^{N}_{n+1}=\frac{e^{2\left(u^{N}_{n+1}-c_N\right)}-1}{e^{2\left(u^{N}_{n}-c_N\right)}-1}\varphi^{N}_{n},
\end{cases} \\
&\begin{cases}
\varphi^{0}_{n,x}=0,\\
\varphi^{j}_{n,x}=\frac{2u^{j}_{n,x}}{e^{2\left(u^{j-1}_{n}-u^{j}_{n}\right)}-1}\left(\varphi^{j-1}_{n}-\varphi^{j}_{n}\right),\quad 1 < j < N-1,\\
\varphi^{N}_{n,x}=\frac{2u^{N}_{n,x}}{e^{2\left(u^{N-1}_{n}-u^{N}_{n}\right)}-1}\left(\varphi^{N-1}_{n}-\varphi^{N}_{n}\right).
\end{cases}
\end{align*}
We put $N=1$ and obtain an equation
\begin{align}\label{sys7N1}
u^{1}_{n+1,x}=u^{1}_{n,x}\frac{\sinh\left(c_0-u^{1}_{n+1}\right)\sinh\left(u^{1}_{n+1}-c_1\right)}{\sinh\left(c_0-u^{1}_{n}\right)\sinh\left(u^{1}_{n}-c_1\right)},
\end{align}
having the Lax pair
\begin{align*}
\begin{cases}
\varphi^{0}_{n+1}=\frac{e^{2\left(c_0-u^{1}_{n}\right)}-1}{e^{2\left(c_0-u^{1}_{n}\right)}-1}\varphi^{0}_{n}+\left(1-\frac{e^{2\left(c_0-u^{1}_{n}\right)}-1}{e^{2\left(c_0-u^{1}_{n}\right)}-1}\right)\varphi^{1}_{n},\\
\varphi^{1}_{n+1}=\frac{e^{2\left(u^{1}_{n+1}-c_1\right)}-1}{e^{2\left(u^{1}_{n}-c_1\right)}-1}\varphi^{1}_{n},
\end{cases} \quad
\begin{cases}
\varphi^{0}_{n,x}=0,\\
\varphi^{1}_{n,x}=\frac{2u^{1}_{n,x}}{e^{2\left(c_0-u^{1}_{n}\right)}-1}\left(\varphi^{0}_{n}-\varphi^{1}_{n}\right).
\end{cases}
\end{align*}
$x$-integrals and $n$-integrals of equation (\ref{sys7N1}) have the form 
\begin{align*} 
J=&\frac{1}{\sinh(c_0-c_1)}\left[arctanh\left(\frac{(\cosh(c_0-c_1)+1)\tanh(\frac{c_0+c_1}{2}-u^{1}_n)}{\sinh(c_0-c_1)}\right)\right.\nonumber\\
&\left.-arctanh\left(\frac{(\cosh(c_0-c_1)+1)\tanh(\frac{c_0+c_1}{2}-u^{1}_{n+1})}{\sinh(c_0-c_1)}\right)\right]
\end{align*}
and, respectively
\begin{align*} 
I=\frac{u^{1}_{n,x}}{\sinh(c_0-u^{1}_n)\sinh(c_1-u^{1}_{n})}.
\end{align*}

\subsection{Scalar equations}

Here we present a list of Darboux integrable scalar equations obtained from the equations (\ref{eq3})--(\ref{eq7}) by imposing degenerate boundary conditions: 
\begin{align*}
&1) \quad u_{n+1,x}=u_{n,x}(u_{n+1})^2;\\
&2) \quad u_{n+1,x}=u_{n,x}\frac{c-u_{n+1}}{u_n};\\
&3) \quad u_{n+1,x}=u_{n,x}\frac{(u_{n+1}+c)(u_{n+1}+c_1)}{u_{n}};\\
&4) \quad u_{n+1,x}=u_{n,x}\frac{u_{n+1}(u_{n+1}-c)}{u_{n}(u_{n}-c)};\\
&5) \quad u_{n+1,x}=u_{n,x}\frac{\sinh (u_{n+1})\sinh(u_{n+1}-c)}{\sinh(u_{n})\sinh(u_{n}-c)}.
\end{align*}
Apparently they are new.

\section*{Conclusions}
The paper proposes an algebraic method for the classification of integrable cases of differential-difference equations of the form (\ref{eq0}). It is shown that all known integrable equations of this form given in the list presented in \cite{FNR} admit reductions as finite systems of differential-difference equations integrable in sense of Darboux. We used this property of the equation (\ref{eq0}) as the basis for the classification criterion. It is convenient to formalize the property of Darboux integrability in terms of its characteristic Lie-Rinehart algebras. The article gives a definition of these algebras for systems of differential-difference equations and investigates their basic properties that may be needed in solving the classification problem.

\section*{Funding}

One of the authors A.R. Khakimova was supported in part by Young Russian Mathematics award.

\end{document}